\documentclass[journal,onecolumn]{IEEEtran}
\IEEEoverridecommandlockouts
\usepackage[T1]{fontenc} 
\usepackage{times}
\usepackage[cmex10]{amsmath}
\usepackage{amssymb}
\usepackage{amsthm}
\usepackage{color}
\usepackage{algorithm}

\usepackage{graphicx}
\usepackage[caption=false]{subfig}
\usepackage{multirow}
\usepackage{algpseudocode}

\usepackage[bookmarks=false,colorlinks=false,pdfborder={0 0 0}]{hyperref}
\usepackage{cite}
\usepackage{bm}
\usepackage{arydshln}
\usepackage{mathtools}
\usepackage{microtype}
\usepackage[figuresright]{rotating}
\usepackage{threeparttable}
\usepackage{booktabs}
\usepackage{verbatim}

\usepackage{amsthm} 
\newcommand{\ppmod}[1]{~({\rm mod~}#1)}
\newtheorem{construction}{Construction}
\newtheorem{example}{Example}
\theoremstyle{definition}
\newtheorem{definition}{Definition}
\usepackage{enumitem}

\newtheorem{theorem}{Theorem}
\newtheorem{proposition}[theorem]{Proposition}
\newtheorem{lemma}[theorem]{Lemma}

\def\BibTeX{{\rm B\kern-.05em{\sc i\kern-.025em b}\kern-.08em
		T\kern-.1667em\lower.7ex\hbox{E}\kern-.125emX}}

\long\def\symbolfootnote[#1]#2{\begingroup
	\def\thefootnote{\fnsymbol{footnote}}\footnote[#1]{#2}\endgroup}
\IEEEoverridecommandlockouts

\title{Redundancy-Optimal Constructions of $(1,1)$-Criss-Cross Deletion Correcting Codes with Efficient Encoding/Decoding Algorithms}

\author{Wenhao Liu, Zhengyi Jiang, Zhongyi Huang, Hanxu Hou
}

\begin{document}
\let\emph\textit
\maketitle
\pagestyle{empty}  
\thispagestyle{empty} 
\begin{abstract}\symbolfootnote[0]{
W. Liu, Z. Jiang and Z. Huang are with 
the Department of Mathematics Sciences, Tsinghua University, Beijing, China~(E-mail: wh-liu24@mails.tsinghua.edu.cn, jzy21@tsinghua.org.cn, zhongyih@tsinghua.edu.cn).
H. Hou is with Dongguan University of Technology, China~(E-mail: houhanxu@163.com). {\em (Corresponding author: Hanxu Hou.)}
This work was partially supported by the National Key R\&D Program of China (No. 2020YFA0712300), the National Natural Science Foundation of China (No. 62071121, 62371411, 12025104),Basic Research Enhancement Program of China under Grant 2021-JCJQ-JJ-0483.}
Two-dimensional error-correcting codes, where codewords are represented as $n \times n$ arrays over a $q$-ary alphabet, find important applications in areas such as QR codes, DNA-based storage, and racetrack memories. 
Among the possible error patterns, $(t_r,t_c)$-criss-cross deletions—where $t_r$ rows and $t_c$ columns are simultaneously deleted—are of particular significance. 
In this paper, we focus on $q$-ary $(1,1)$-criss-cross deletion correcting codes. 
We present a novel code construction and develop complete encoding, decoding, and data recovery algorithms for parameters $n \ge 11$ and $q \ge 3$. 
The complexity of the proposed encoding, decoding, and data recovery algorithms is $\mathcal{O}(n^2)$. 
Furthermore, we show that for $n \ge 11$ and $q = \Omega(n)$ (i.e., there exists a constant $c>0$ such that $q \ge cn$), both the code redundancy and the encoder redundancy of the constructed codes are
$2n + 2\log_q n + \mathcal{O}(1),$
which attain the lower bound ($2n + 2\log_q n - 3$) within an $\mathcal{O}(1)$ gap. 
To the best of our knowledge, this is the first construction that can achieve the optimal redundancy with only an $\mathcal{O}(1)$ gap, while simultaneously featuring explicit encoding and decoding algorithms.

\end{abstract}

\begin{IEEEkeywords}
Criss-cross deletions, explicit coding algorithms, code redundancy, encoder redundancy
\end{IEEEkeywords}

\IEEEpeerreviewmaketitle

\section{Introduction}
Two-dimensional error-correcting codes, whose codewords are stored in the form of two-dimensional arrays, play an important role in various applications such as QR codes~\cite{QR_1,QR_2}, DNA-based storage~\cite{DNA1}, and racetrack memories~\cite{rack1,rack2,rack3,rack4}. 
Existing studies on two-dimensional error-correcting codes can be broadly classified into four categories, namely, coding schemes designed to correct substitution errors~\cite{Existing1,Existing2,Existing3,Existing4,Existing5,Existing6,Existing7,Existing10,Existing11,Existing12}, erasure errors~\cite{Existing2,Existing12}, deletion errors~\cite{Existing8,Existing9,Existing12,Existing_Crisscross1,Existing_Crisscross2, 2021TIT_crisscross1,2021ISIT_crisscross2,2025arxiv_crisscross3,2022TIT_multi}, and insertion errors~\cite{Existing_Crisscross1, 2021TIT_crisscross1,2021ISIT_crisscross2,2025arxiv_crisscross3,2022TIT_multi}. This work focuses on the case of deletion errors.

Among the possible deletion patterns in two-dimensional error-correcting codes, \emph{criss-cross} deletions are among the most representative. Specifically, a \emph{\((t_r,t_c)\)-criss-cross deletion} refers to the removal of \(t_r\) entire rows and \(t_c\) entire columns from the codeword array, where $t_r, t_c \ge 1$ are integers. Let \(n\ge 1\) and \(q\ge 2\) be integers. A \emph{\(q\)-ary \((t_r,t_c)\)-criss-cross deletion correcting code} \(\mathcal{C} \subseteq \Sigma_q^{n\times n}\) is a set of codewords, each represented as an \(n\times n\) array over the alphabet \(\Sigma_q \triangleq \{0,1,\ldots,q-1\}\), such that any \((t_r,t_c)\)-criss-cross deletion can be corrected to recover the original symbols. Several previous works have investigated coding schemes against criss-cross deletion errors, highlighting their theoretical and practical significance~\cite{Existing_Crisscross1,Existing_Crisscross2,2021TIT_crisscross1,2021ISIT_crisscross2,2025arxiv_crisscross3,2022TIT_multi}.

The complete encoding and decoding procedure of a \(q\)-ary \((t_r,t_c)\)-criss-cross deletion correcting code $\mathcal{C}$ consists of the following steps:
\begin{enumerate}
    \item \label{step:msg} Given \( k \) data symbols over \(\Sigma_q\), where \(1 \le k \le n^2\), the encoding algorithm maps them to a codeword \(\bm{X} \in \mathcal{C} \subseteq \Sigma_q^{n\times n}\);
    \item When a \((t_r,t_c)\)-criss-cross deletion occurs, the resulting array is \(\bm{Y} \in \Sigma_q^{(n-t_r) \times (n-t_c)}\);
    \item Given only \(\bm{Y}\), the decoding algorithm reconstructs the original codeword \(\bm{X}\);
    \item Finally, the data recovery algorithm recovers the original \(k\) data symbols from \(\bm{X}\).
\end{enumerate}

For \(q\)-ary \((t_r,t_c)\)-criss-cross deletion correcting codes, one of the key metrics is \emph{redundancy}. In the existing literature, redundancy has been defined in two different ways. To avoid confusion, we summarize their definitions as follows.

\textbf{Code redundancy}~\cite{Existing_Crisscross1,2021TIT_crisscross1,2021ISIT_crisscross2} is defined as
\begin{equation*}
    r_{\mathcal{C}} \triangleq n^2 - \lfloor \log_q |\mathcal{C}| \rfloor,
\end{equation*} where $|\mathcal{C}|$ denotes the cardinality (i.e., the number of codewords) of the code $\mathcal{C}$.
This quantity depends only on the code size and characterizes the fundamental limit on the amount of data that can be carried by the codewords, without requiring an explicit encoding or decoding algorithm. Specifically, the maximum number of \(q\)-ary data symbols that can be carried by \(\mathcal{C}\) is $
    k^\star \triangleq \lfloor \log_q |\mathcal{C}| \rfloor.$

\textbf{Encoder redundancy}~\cite{2021TIT_crisscross1,2024Differential_VT_codes} is defined as follows:  
If a \(q\)-ary \((t_r,t_c)\)-criss-cross deletion correcting code \(\mathcal{C}\) is given together with explicit encoding, decoding, and data recovery algorithms, 
then the encoder redundancy of this code (with respect to the given implementation) is defined as
\begin{equation*}
    r_{\mathrm{ENC}} \triangleq n^2 - k,
\end{equation*}where \(k\) is the number of data symbols.
Since an explicit encoding procedure that maps \(k\) data symbols to codewords can generate \(q^k\) distinct codewords, there must be $q^k \le |\mathcal{C}|$,
which implies that
\begin{equation*}
    k \le \lfloor \log_q |\mathcal{C}| \rfloor = k^\star .
\end{equation*}
Therefore,
\begin{equation*}
    r_{\mathrm{ENC}} \ge n^2 - k^\star = r_{\mathcal{C}} .
\end{equation*}

Consequently, specifying an encoder redundancy \(r_{\mathrm{ENC}}\) for a code immediately implies that the corresponding code redundancy \(r_{\mathcal{C}}\) satisfies \(r_{\mathcal{C}} \le r_{\mathrm{ENC}}\). Moreover, it requires providing explicit encoding, decoding, and data recovery algorithms that achieve this redundancy. Hence, from both theoretical and practical perspectives, constructing codes with small encoder redundancy is more challenging and of greater significance.

For \(q\)-ary \((t_r=1,t_c=1)\)-criss-cross deletion correcting codes, the case that this paper focuses on, it was shown in~\cite{2021TIT_crisscross1} that the code redundancy has a lower bound of
\begin{equation*}
    2n + 2\log_q n - 3.
\end{equation*}
 As the encoder redundancy is always at least as large as the code redundancy, it is subject to the same lower bound. Existing works that attempt to approach this lower bound can be divided into two categories:

\begin{itemize}
    \item \emph{Code redundancy}: these works do not provide explicit encoding algorithms, but only code constructions together with decoding algorithms~\cite{Existing_Crisscross1,2021TIT_crisscross1,2021ISIT_crisscross2,2025arxiv_crisscross3}. The best known result is given in~\cite{2021ISIT_crisscross2}, achieving a code redundancy of
    \[
        2n + 2\log_q n + \mathcal{O}(\log\log n).
    \]
    \item \emph{Encoder redundancy}: the only known work in this category is~\cite{2021TIT_crisscross1}, which provides explicit encoding, decoding, and data recovery algorithms. However, its encoder redundancy is not tight, and moreover, as will be clarified in the forthcoming \textbf{Remark}, the decoding procedure described therein cannot handle all possible $(1,1)$-criss-cross deletion patterns in certain cases. Consequently, to the best of our knowledge, there is currently no known construction whose encoder redundancy is both tight and capable of correcting all types of $(1,1)$-criss-cross deletion patterns.

\end{itemize}

\textbf{Remark.} A construction achieving code redundancy of \(2n + 2\log_q n + \mathcal{O}(1)\) was proposed in~\cite{2025arxiv_crisscross3}, but, as discussed in Appendix~\ref{appendix:2025problem}, its decoding algorithm exhibits certain ambiguities that may pose challenges for direct implementation.
On the other hand, the encoder redundancy presented in~\cite{2021TIT_crisscross1} is \(2n + 9\log_q n + 12 + 2\log_2 e\), where \(e\) denotes the base of the natural logarithm. However, as noted in Appendix~\ref{appendix:2021TIT_problem}, the decoding procedure therein may not succeed for all deletion patterns and therefore requires further refinement. 
In this paper, we provide the first explicit construction that achieves correct decoding for all $(1,1)$-criss-cross deletion patterns together with optimal encoder redundancy (within an $\mathcal{O}(1)$ gap to the lower bound).

Our main contributions are summarized as follows:

\begin{enumerate}
    \item For $q$-ary $(t_r=1,t_c=1)$-criss-cross deletion correcting codes, we propose a novel code construction together with a complete encoding, decoding, and data recovery procedure for the parameter range $n \ge 11$ and $q \ge 3$. The proposed algorithms achieve an encoding and decoding complexity of $\mathcal{O}(n^2)$. To the best of our knowledge, this is the first construction of $(1,1)$-criss-cross deletion correcting codes that provides explicit encoding and decoding algorithms.
    
    \item For the proposed construction, when $n \ge 11$ and $q = \Omega(n)$ (i.e., there exists a constant $c>0$ such that $q \ge c n$), both the code redundancy and the encoder redundancy are
$$
2n + 2\log_q n + \mathcal{O}(1),
$$
which shows that both redundancies reach the lower bound within an $\mathcal{O}(1)$ gap. In particular, 
this construction surpasses the best-known redundancy of $2n + 2\log_q n + \mathcal{O}(\log\log n)$~\cite{2021ISIT_crisscross2}. 

\end{enumerate}

\section{Preliminaries}
In this section, we first recall 
the $q$-ary Differential VT codes~\cite{2024Differential_VT_codes}, which are important one-dimensional codes capable of correcting a single deletion or a single insertion  and will serve as building blocks in the construction of our two-dimensional codes.
Then, we review relevant results on $(1,1)$-criss-cross deletion correcting codes~\cite{2021TIT_crisscross1,2025arxiv_crisscross3}.

We begin by defining some notations. Let $|\cdot|$ denote the cardinality of a set. For positive integers $q\ge 2$ and $n$, let $\Sigma_q = \{0,1,\ldots,q-1\}$ denote the $q$-ary alphabet, and let $\Sigma_q^n$ denote the set of all sequences of length $n$ over $\Sigma_q$. Denote by $\mathbb{Z}_n$ the ring of integers modulo $n$, and for integers $i \le j$, define $[i:j] = \{i, i+1, \ldots, j\}$. Let $\Sigma_q^{n \times n}$ denote the set of all $n\times n$ arrays with entries in $\Sigma_q$.

For a vector $\bm{x} = (x_1, x_2, \ldots, x_n) \in \Sigma_q^n$ and integers $i_1 \le i_2$, denote by $
\bm{x}_{[i_1:i_2]} \triangleq (x_{i_1}, x_{i_1+1}, \ldots, x_{i_2})
$ the subsequence of $\bm{x}$ from position $i_1$ to $i_2$. For an array $\bm{X} \in \Sigma_q^{n\times n}$, let $\bm{X}_{i,j}$ denote the element in the $i$-th row and the $j$-th column. Denote by $\bm{X}_{i,[j_1:j_2]}$ the sequence formed by the elements in the $i$-th row from column $j_1$ to $j_2$, and by $\bm{X}_{[i_1:i_2],j}$ the sequence formed by the elements in the $j$-th column from row $i_1$ to $i_2$. Finally, let $\bm{X}_{[i_1:i_2],[j_1:j_2]}$ denote the subarray consisting of rows $i_1$ to $i_2$ and columns $j_1$ to $j_2$ of $\bm{X}$.

\subsection{$q$-ary Differential VT Codes}
\label{subsec:pre_one_dimension}

Before introducing $q$-ary Differential VT codes, we first formalize the definitions of a single deletion and a single insertion in a sequence~\cite{1965VTcode}.

Given a sequence $\bm{x}=(x_1,\ldots,x_n)\in \Sigma_q^n$:

\begin{itemize}
    \item $\bm{x}$ \emph{is said to suffer a single deletion} if exactly one symbol is removed from $\bm{x}$. Formally, there exists $1 \le i \le n$ such that the resulting sequence is
    $$
    (x_1,\ldots,x_{i-1},x_{i+1},\ldots,x_n).
    $$
    
    \item $\bm{x}$ \emph{is said to suffer a single insertion} if exactly one symbol is inserted into $\bm{x}$. Formally, there exist $1 \le i \le n+1$ and $y \in \Sigma_q$ such that the resulting sequence is
    $$
    (x_1,\ldots,x_{i-1},y,x_i,\ldots,x_n).
    $$
\end{itemize}

We then review the definitions of the differential vector and the VT syndrome~\cite{2024Differential_VT_codes}, which are used to construct the $q$-ary Differential VT codes.

For $\bm{x}=(x_1,x_2,\ldots,x_n)\in \Sigma_q^n$, the \emph{differential vector} of $\bm{x}$ is denoted by $\bm{y}=\mathrm{Diff}(\bm{x})=(y_1,y_2,\ldots,y_n)\in \Sigma_q^n$, where
\[
y_i =
\begin{cases}
x_i - x_{i+1} \ppmod q, & i=1,2,\ldots,n-1, \\[6pt]
x_n, & i=n .
\end{cases}
\]

This mapping is bijective, and its inverse $\bm{x}=\mathrm{Diff}^{-1}(\bm{y})$ is given by

$$
x_i \;=\;
\begin{cases}
\displaystyle \sum_{j=i}^n y_j\ppmod{q}, & i=1,2,\ldots,n-1, \\[8pt]
y_n, & i=n .
\end{cases}
$$

The \emph{VT syndrome} of a $q$-ary sequence $\bm{x}=(x_1,x_2,\ldots,x_n)\in \Sigma_q^n$ is defined as

$$
\mathrm{Syn}(\bm{x}) \triangleq \sum_{i=1}^n i x_i.
$$

Given $a \in \mathbb{Z}_{qn}$, the $q$-ary Differential VT codes with length $n$ (denoted as $\mathrm{Diff\_VT}_{a}(n;q)$) \cite{2024Differential_VT_codes} is defined as

\begin{align*}
&\mathrm{Diff\_VT}_{a}(n;q)
\triangleq \Bigl\{ \bm{x}\in \Sigma_q^n : 
\mathrm{Syn}(\mathrm{Diff}(\bm{x})) \equiv a \ppmod{qn} \Bigr\}.
\end{align*}

Two fundamental properties of Differential VT codes are summarized below.

\begin{lemma}{\cite[Lemma~2]{2024Differential_VT_codes}}\label{lemma:DVT_lemma}
Given $n\ge1$, $q\ge2$, and $a\in \mathbb{Z}_{qn}$, for any $\bm{x}\in \mathrm{Diff\_VT}_{a}(n;q)$, it holds that
$     \sum_{i=1}^n x_i \equiv a \ppmod{q}.
    $
\end{lemma}

\begin{theorem}{\cite[Theorem~3]{2024Differential_VT_codes}}\label{theorem:2024DVT}
Given any codeword $\bm{x} \in \mathrm{Diff\_VT}_{a}(n;q)$, if $\bm{x}$ suffers a single deletion or a single insertion, there exists a linear-time decoding algorithm that can uniquely reconstruct the original codeword $\bm{x}$.
\end{theorem}

According to Theorem~\ref{theorem:2024DVT}, the Differential VT codes guarantee recovery of the original codeword under a single deletion, although the precise deletion position may remain ambiguous. 
As shown in~\cite{2024Differential_VT_codes}, by imposing an additional constraint on the codewords, one can ensure that not only the original codeword but also the exact deletion position can be uniquely located, as stated in the following proposition.

\begin{proposition}{\cite{2024Differential_VT_codes}}\label{prop:unique_deletion_position}
Let $\bm{x}=(x_1,x_2,\ldots,x_n)\in \mathrm{Diff\_VT}_{a}(n;q)$ be such that adjacent symbols are distinct, i.e., $x_i \neq x_{i+1}$ for all $i=1,2,\ldots,n-1$.
If $\bm{x}$ suffers a single deletion, then the original sequence $\bm{x}$ can be uniquely recovered, and the deletion position can also be identified.
\end{proposition}

\subsection{$(1,1)$-Criss-Cross Deletion Correcting Codes}
\label{subsec:pre_two_dimension}

This subsection introduces the $(1,1)$-criss-cross deletion correcting codes~\cite{2021TIT_crisscross1, 2025arxiv_crisscross3}.

Given $\bm{X} \in \Sigma_q^{n\times n}$, i.e.,

$$
\bm{X} = 
\begin{bmatrix}
    X_{1,1} & X_{1,2} & \cdots & X_{1,n} \\
    X_{2,1} & X_{2,2} & \cdots & X_{2,n} \\
    \vdots  & \vdots  & \ddots & \vdots \\
    X_{n,1} & X_{n,2} & \cdots & X_{n,n}
\end{bmatrix},
$$
$\bm{X}$ is said to suffer a \emph{$(1,1)$-criss-cross deletion} at location $(i,j)$, where $1\le i,j\le n$, if the $i$-th row and $j$-th column of $\bm{X}$ are deleted. The resulting $(n-1)\times(n-1)$ array is

$$
\begin{bmatrix}
    \bm{X}_{[1:i-1],[1:j-1]} & \bm{X}_{[1:i-1],[j+1:n]} \\
    \bm{X}_{[i+1:n],[1:j-1]} & \bm{X}_{[i+1:n],[j+1:n]}
\end{bmatrix}.
$$

The set of all possible arrays obtained from $\bm{X}$ by a $(1,1)$-criss-cross deletion is called the \emph{$(1,1)$-criss-cross deletion ball} of $\bm{X}$, denoted by $\mathbb{D}_{1,1}(\bm{X})$.

A code $\mathcal{C} \subseteq \Sigma_q^{n\times n}$ is called a \emph{$q$-ary $(1,1)$-criss-cross deletion correcting code} if, for any $\bm{X} \neq \bm{Y} \in \mathcal{C}$,

$$
\mathbb{D}_{1,1}(\bm{X}) \cap \mathbb{D}_{1,1}(\bm{Y}) = \emptyset.
$$

It is straightforward to see that a code defined in this way can correct a $(1,1)$-criss-cross deletion. A \emph{$q$-ary $(1,1)$-criss-cross insertion correcting code} can be defined similarly; see~\cite{2021TIT_crisscross1, 2025arxiv_crisscross3} for details. In the following, we omit the prefix “$q$-ary” when it does not cause ambiguity.

Lemma~\ref{lemma:del_insert_equivalence} shows that a code capable of correcting $(1,1)$-criss-cross deletions can also correct $(1,1)$-criss-cross insertions. Therefore, we focus exclusively on $(1,1)$-criss-cross deletion correcting codes.

\begin{lemma}{\cite[Lemma~4]{2025arxiv_crisscross3}}\label{lemma:del_insert_equivalence}
A code $\mathcal{C} \subseteq \Sigma_q^{n \times n}$ is a $(1,1)$-criss-cross deletion correcting code if and only if it is a $(1,1)$-criss-cross insertion correcting code.
\end{lemma}

The following theorem provides a lower bound on the code redundancy of $(1,1)$-criss-cross deletion correcting codes.

\begin{theorem}{\cite[Theorem~11]{2021TIT_crisscross1}}\label{theorem:theoretic_bound}
Let $r_{\mathcal{C}}$ denote the code redundancy of a $q$-ary $(1,1)$-criss-cross deletion correcting code $\mathcal{C}$. Then

$$
r_{\mathcal{C}} \ge 2n + 2 \log_q(n) - 3 .
$$

\end{theorem}

\section{Construction and Decoding Algorithm}\label{sec:construction}

In this section, we present the construction of our $q$-ary $(1,1)$-criss-cross deletion correcting codes, and provide the corresponding decoding algorithm.

\subsection{Construction}

\begin{figure}
    \centering
    \includegraphics[width=.6\linewidth]{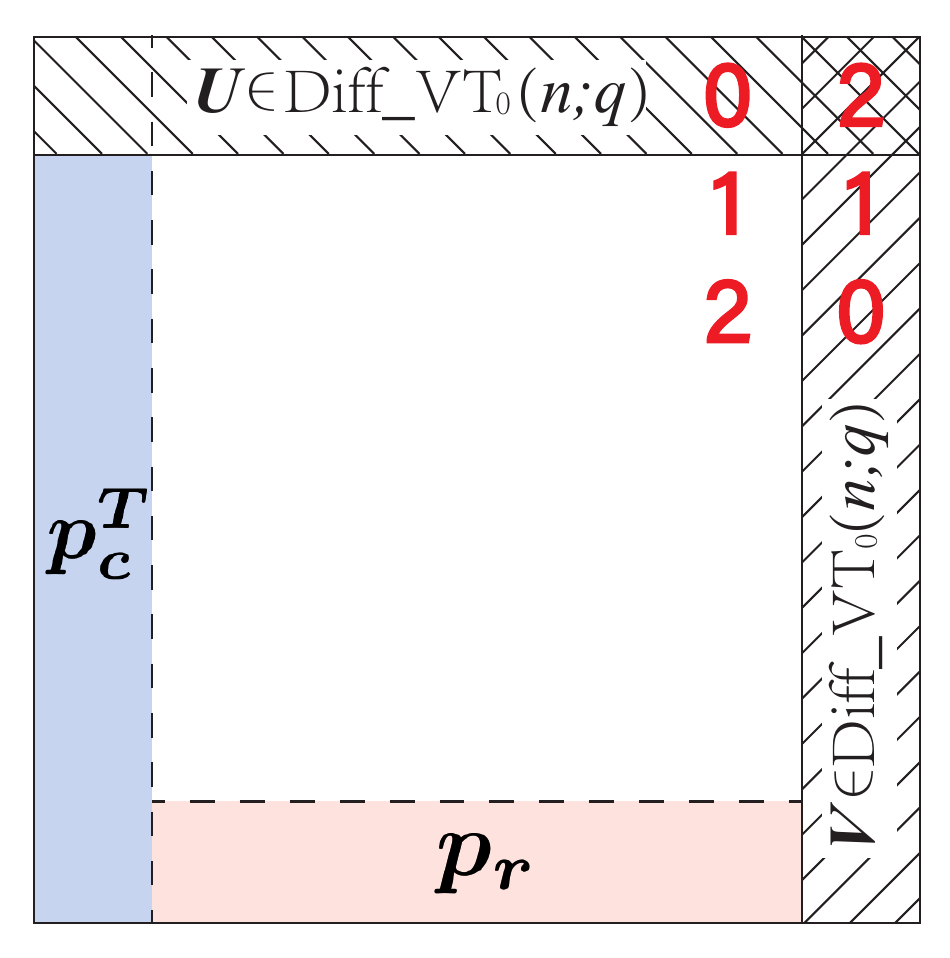}
    \caption{Illustration of the structure of the proposed $q$-ary $(1,1)$-criss-cross deletion correcting code. The first row and the reversed last column are encoded using differential VT codes, and specific entries and modular sum constraints are imposed to ensure unique recovery from any $(1,1)$-criss-cross deletion.}
    \label{fig:1}
\vspace{-0.5cm}
	\end{figure}

We give the construction of our code as follows, which is illustrated in Fig.~\ref{fig:1}.

\begin{construction}\label{constrcution:ourcode}
 Given integers $q \ge 3$, $n \ge 4$, and the $q$-ary alphabet $\Sigma_q$. The $q$-ary $(1,1)$-criss-cross deletion correcting code $\mathcal{C}$ is defined as the set of all $n \times n$ arrays $\bm{X} \in \Sigma_q^{n\times n}$ satisfying the following five conditions.

\begin{enumerate}
\item The first row of $\bm{X}$, denoted by $\bm{U} = (u_1, \ldots, u_n)$ with $u_i = \bm{X}_{1,i}$ for $i = 1,2,\ldots,n$, is required to satisfy $\bm{U} \in \mathrm{Diff\_VT}_0(n;q)$, and
\[
u_i \neq u_{i+1} \ \text{for all } i=1,2,\ldots,n-1,
\]
together with the suffix constraints
\[
u_{n-1} = 0, \quad u_n = 2.
\]

\item The reversal of the last column of $\bm{X}$, denoted by $\bm{V} = (v_1, \ldots, v_n)$ with $v_i = \bm{X}_{n-i+1,n}$ for $i = 1,2,\ldots,n$, is required to satisfy $\bm{V} \in \mathrm{Diff\_VT}_0(n;q)$, and
\[
v_i \neq v_{i+1} \ \text{for all } i=1,2,\ldots,n-1,
\]
together with the suffix constraints
\[
v_{n-2} = 0, \quad v_{n-1} = 1, \quad v_n = 2.
\]

\item The entries are required to satisfy the fixed values
\[
\bm{X}_{2,n-1} = 1, \quad \bm{X}_{3,n-1} = 2.
\]

\item The vector $\bm{p_c} \triangleq (s_2, s_3, \ldots, s_n)$, where $s_i = \bm{X}_{i,1}$ for $i=2,3,\ldots,n$, is required to satisfy
\[
s_i + \sum_{j=2}^n \bm{X}_{i,j} \equiv 0 \ppmod{q}, \quad i=2,3,\ldots,n.
\]

\item The vector $\bm{p_r}\triangleq (t_2, t_3, \ldots, t_{n-1})$, where $t_i = \bm{X}_{n,i}$ for $i=2,3,\ldots,n-1$, is required to satisfy
\[
t_i + \sum_{j=1}^{n-1} \bm{X}_{j,i} \equiv 0 \ppmod{q}, \quad i=2,3,\ldots,n-1.
\]

\end{enumerate}
\end{construction}

Next, we present an example to illustrate the new construction, which will also serve as a running example throughout the subsequent encoding and decoding procedures.

\begin{example}\label{example:encoding}
Consider the case with parameters $n=9$ and $q=7$. 
Let 
\[\bm{X}=
\begin{bmatrix}
4 & 2 & 1 & 4 & 5 & 2 & 1 & 0 & 2 \\
5 & 0 & 2 & 4 & 4 & 0 & 4 & 1 & 1 \\
5 & 1 & 2 & 6 & 0 & 3 & 2 & 2 & 0 \\
5 & 6 & 1 & 0 & 3 & 2 & 4 & 6 & 1 \\
2 & 4 & 3 & 5 & 6 & 1 & 6 & 1 & 0 \\
3 & 1 & 3 & 2 & 3 & 2 & 2 & 6 & 6 \\
5 & 0 & 0 & 5 & 0 & 5 & 5 & 3 & 5 \\
3 & 0 & 6 & 2 & 3 & 3 & 0 & 5 & 6 \\
3 & 0 & 3 & 0 & 4 & 3 & 4 & 4 & 0
\end{bmatrix}.
\]

It is easy to verify that the first row
\[
\bm{U} = (4,2,1,4,5,2,1,0,2),
\]
belongs to $\mathrm{Diff\_VT}_0(9;7)$, satisfies the constraint $u_i \neq u_{i+1}$ for all $i\in\{1,2,\ldots,8\}$, and also satisfies the suffix constraints $u_{8}=0$, $u_{9}=2$.  

Similarly, the reversal of the last column is
\[
\bm{V} = (0,6,5,6,0,1,0,1,2),
\]
which belongs to $\mathrm{Diff\_VT}_0(9;7)$, satisfies the constraint $v_i \neq v_{i+1}$ for all $i\in\{1,2,\ldots,8\}$, and satisfies the suffix constraints $v_{7}=0$, $v_{8}=1$, $v_{9}=2$.  

Next, the fixed entries condition is satisfied since
\[
\bm{X}_{2,8} = 1, \quad \bm{X}_{3,8} = 2.
\]

Moreover, the vector $\bm{p_c} = (s_2, \ldots, s_9) = (5,5,5,2,3,5,3,3)$ satisfies
\[
s_i + \sum_{j=2}^9 \bm{X}_{i,j} \equiv 0 \ppmod{7}, \quad i=2,3,\ldots,9,
\]
and the vector $\bm{p_r} = (t_2, \ldots, t_{8}) = (0,3,0,4,3,4,4)$ satisfies
\[
t_i + \sum_{j=1}^8 \bm{X}_{j,i} \equiv 0 \ppmod{7}, \quad i=2,3,\ldots,8.
\]

Therefore, the array $\bm{X}$ indeed belongs to the codeword of \textbf{Construction~\ref{constrcution:ourcode}}. 
\end{example}

\subsection{Decoding Algorithm}\label{subsec:2Ddecoding}

In this section, we present the decoding algorithm for \textbf{Construction~\ref{constrcution:ourcode}}. To this end, we first establish the following lemma.

\begin{lemma}\label{lemma:sum0}
For any codeword $\bm{X} \in \mathcal{C}$ in \textbf{Construction~\ref{constrcution:ourcode}},

\begin{align*}
\sum_{i=1}^n \bm{X}_{i,j} &\equiv 0 \ppmod{q}, \quad \forall j = 1,2,\ldots,n, \\
\sum_{j=1}^n \bm{X}_{i,j} &\equiv 0 \ppmod{q}, \quad \forall i = 1,2,\ldots,n.
\end{align*}
\end{lemma}

\begin{proof}
Since the first row $\bm{U}$ of $\bm{X}$ belongs to $\mathrm{Diff\_VT}_0(n;q)$, we have $\sum_{j=1}^n \bm{X}_{1,j} \equiv 0 \ppmod{q}$ by Lemma~\ref{lemma:DVT_lemma}.
Similarly, since $V \in \mathrm{Diff\_VT}_0(n;q)$, we have
$
\sum_{i=1}^n \bm{X}_{i,n} \equiv 0 \ppmod{q}.
$

Then, by conditions (4) and (5) in \textbf{Construction~\ref{constrcution:ourcode}}, it follows that

$$
\sum_{j=1}^n \bm{X}_{i,j} \equiv 0 \ppmod{q}, \quad \forall i = 2,3,\ldots,n,
$$

and

$$
\sum_{i=1}^n \bm{X}_{i,j} \equiv 0 \ppmod{q}, \quad \forall j = 2,3,\ldots,n-1.
$$

Finally, note that

$$
\sum_{j=1}^n \Bigl(\sum_{i=1}^n \bm{X}_{i,j}\Bigr) \equiv \sum_{i=1}^n \Bigl(\sum_{j=1}^n \bm{X}_{i,j}\Bigr) \equiv 0 \ppmod{q},
$$
where the second equality makes use of the row-sum constraints, we obtain

$$
\sum_{i=1}^n \bm{X}_{i,1} \equiv 0 - \sum_{j=2}^{n} \Bigl(\sum_{i=1}^{n} \bm{X}_{i,j}\Bigr) \equiv 0 \ppmod{q}.
$$

Combining all the above, the lemma is proved.
\end{proof}

When a codeword $\bm{X} \in \mathcal{C}$ suffers a $(1,1)$-criss-cross deletion, the resulting array is denoted by $\bm{Y} \in \Sigma_q^{(n-1) \times (n-1)}$. The decoding algorithm aims to recover the original codeword $\bm{X}$ from $\bm{Y}$. Algorithm~\ref{alg:decoding} presents the decoding procedure for \textbf{Construction~\ref{constrcution:ourcode}}, where $\bm{x}^T$ denotes the transpose of the vector $\bm{x}$ and $\bm{Y}_{i,j}$ represents the entry in the $i$-th row and $j$-th column of $\bm{Y}$. The correctness of this algorithm and its computational complexity are established in Theorem~\ref{theorem:decoding_correct}.

\begin{algorithm}
\caption{Decoding Algorithm for \textbf{Construction~\ref{constrcution:ourcode}}}
\label{alg:decoding}
\begin{algorithmic}[1]
\Require Received array $\bm{Y} \in \Sigma_q^{(n-1) \times (n-1)}$
\Ensure Original codeword $\bm{X} \in \mathcal{C}$
\If {$\bm{Y}_{1,n-1} < \bm{Y}_{2,n-1}$}
\State Extend $\bm{Y}$ as
$         \bm{Y} \leftarrow \begin{bmatrix} \bm{Y} & \bm{y} \end{bmatrix},
    $
where $\bm{y} = (y_1,\ldots,y_{n-1})^T$ with
$         y_i = -\sum_{j=1}^{n-1} \bm{Y}_{i,j} \ppmod q$, for $i = 1,2,\ldots,n-1$.
\EndIf
\State Denote by $\bm{V}' = (v_1, \ldots, v_{n-1})$ the reversal of the last column of $\bm{Y}$. By \textbf{Construction~\ref{constrcution:ourcode}}, $\bm{V}'$ is obtained from $\bm{V} \in \mathrm{Diff\_VT}_0(n;q)$ by a single deletion, where all adjacent symbols in $\bm{V}$ are distinct. Hence, by Proposition~\ref{prop:unique_deletion_position}, we can apply the decoding algorithm of the Differential VT codes, which uniquely determines the position of the deletion. Denote the corresponding row index by $i^*$.

\State Let $n'$ be the current number of columns of $\bm{Y}$, and compute $\bm{r} = (r_1, \ldots, r_{n'})$, where

$$
    r_j = -\sum_{i=1}^{n-1} \bm{Y}_{i,j} \ppmod q, \quad j = 1,2,\ldots,n'.
$$

\State Insert $\bm{r}$ as the $i^*$-th row of $\bm{Y}$:

$$
    \bm{Y} \leftarrow
    \begin{bmatrix}
        \bm{Y}_{[1:i^*-1], [1:n']} \\
        \bm{r} \\
        \bm{Y}_{[i^*:n-1], [1:n']}
    \end{bmatrix}.
$$

\If {$n' = n-1$}
\State Denote the first row of $\bm{Y}$ be $\bm{U}' = (u_1, \ldots, u_{n-1})$. By \textbf{Construction~\ref{constrcution:ourcode}}, $\bm{U}'$ is obtained from $\bm{U} \in \mathrm{Diff\_VT}_0(n;q)$ by a single deletion, with adjacent symbols distinct. By Proposition~\ref{prop:unique_deletion_position}, we can apply the decoding algorithm of the Differential VT codes to determine the position of the deletion, and denote the column index by $j^*$.
\State Compute
$         \bm{c} = (c_1, \ldots, c_n)$, where $$c_i = -\sum_{j=1}^{n-1} \bm{Y}_{i,j} \ppmod q, \quad i = 1,2,\ldots,n.$$
\State Insert $\bm{c}$ as the $j^*$-th column of $\bm{Y}$:
$$         \bm{Y} \leftarrow
        \begin{bmatrix}
            \bm{Y}_{[1:n], [1:j^*-1]} & \bm{c}^T & \bm{Y}_{[1:n], [j^*:n-1]}
        \end{bmatrix}.
    $$
\EndIf
\State \Return $\bm{Y}$
\end{algorithmic}
\end{algorithm}

\begin{theorem}\label{theorem:decoding_correct}
Algorithm~\ref{alg:decoding} correctly reconstructs any codeword 
$\bm{X} \in \mathcal{C}$ from its deleted array $\bm{Y}$, with 
complexity $\mathcal{O}(n^2)$.
\end{theorem}

\begin{proof}
We first prove that Steps~1--3 of Algorithm~\ref{alg:decoding} correctly determine whether the last column was deleted and, if so, recover $n-1$ symbols of the deleted column (all except the symbol in the deleted row, which cannot yet be recovered). Consider all possible values of the pair $(\bm{Y}_{1,n-1}, \bm{Y}_{2,n-1})$ under different deletion patterns. One can verify that all possible cases are $(2,1),(2,0),(1,0),(0,1),(0,2),(1,2)$. The first three correspond to the case where the last column was not deleted, while the latter three correspond to the case where the last column was deleted. Hence, the condition $\bm{Y}_{1,n-1} < \bm{Y}_{2,n-1}$ holds if and only if the last column was deleted. If this is the case,  Step~2 correctly reconstructs all entries of the last column except the one located in the deleted row by Lemma~\ref{lemma:sum0}. 

Next, we show that Steps~4–6 correctly identify the index of the deleted row and recover the $n'$ symbols of that row, where $n' \in \{n-1, n\}$. At the beginning of Step~4, regardless of whether the last column was deleted, the last column of the current array (after the possible extension in Step~2) coincides with the last column of $\bm{X}$ with exactly one symbol removed—namely, the entry located in the deleted row. Let $\bm{V}'=(v_1,\ldots,v_{n-1})$ be the reversal of this column. By \textbf{Construction~\ref{constrcution:ourcode}}, $\bm{V}$, defined as the reversal of the last column of $\bm{X}$, lies in $\mathrm{Diff\_VT}_0(n;q)$ and has all adjacent symbols distinct, and $\bm{V}'$ is obtained from $\bm{V}$ by a single deletion. Hence, by Proposition~\ref{prop:unique_deletion_position}, the decoding algorithm of the Differential VT codes uniquely locates the position of the deletion, which in turn yields the index $i^*$ of the deleted row in the two-dimensional code. Steps~5--6 then reconstructs the symbols of row $i^*$ using the modulo-$q$ sum constraints in Lemma~\ref{lemma:sum0} and inserts the recovered row. If the last column was the one deleted, Step~2 has already reconstructed $n-1$ entries of that column; knowing $i^*$ allows us to fill the remaining entry as well, so the entire row is recovered (i.e., $n$ symbols). Otherwise, when a different column was deleted, all entries of row $i^*$ except the one in the deleted column are recovered (i.e., $n-1$ symbols).

At this point, two situations may arise. If the last column was deleted, then $\bm{X}$ has already been fully recovered, corresponding to the case $n'=n$, and no further action is needed. Otherwise, when another column was deleted, the deleted row has been partially recovered (with $n-1$ symbols known), corresponding to the case $n'=n-1$. In this case, it remains to recover the deleted column. Since the row deletion has already been corrected, the first row of $\bm{Y}$ is obtained from the first row of $\bm{X}$ by a single deletion. Similarly, by applying the decoding algorithm of the Differential VT codes, we can identify the index $j^*$ of the deleted column. Lemma~\ref{lemma:sum0} then allows us to reconstruct the entire deleted column. Steps~8–10 implement these operations.

In conclusion, Algorithm~\ref{alg:decoding} correctly recovers both the deleted row and the deleted column. Since the complexity of decoding the one-dimensional Differential VT codes is $\mathcal{O}(n)$~\cite{2024Differential_VT_codes}, while the complexity of the remaining operations is $\mathcal{O}(n^2)$, the overall complexity of the algorithm is $\mathcal{O}(n^2)$. This completes the proof.
\end{proof}

We illustrate the decoding procedure for the code in Example~\ref{example:encoding}. 

Suppose that the codeword $\bm{X}$ suffers a $(1,1)$-criss-cross deletion at position $(9,9)$, i.e., the 9-th row and 9-th column are deleted, resulting in the array
\[
\bm{Y} = 
\begin{bmatrix}
4 & 2 & 1 & 4 & 5 & 2 & 1 & 0 \\
5 & 0 & 2 & 4 & 4 & 0 & 4 & 1 \\
5 & 1 & 2 & 6 & 0 & 3 & 2 & 2 \\
5 & 6 & 1 & 0 & 3 & 2 & 4 & 6 \\
2 & 4 & 3 & 5 & 6 & 1 & 6 & 1 \\
3 & 1 & 3 & 2 & 3 & 2 & 2 & 6 \\
5 & 0 & 0 & 5 & 0 & 5 & 5 & 3 \\
3 & 0 & 6 & 2 & 3 & 3 & 0 & 5
\end{bmatrix}.
\]
We now demonstrate how to apply Algorithm~\ref{alg:decoding} to recover the original codeword $\bm{X}$. 

In Step~1, we verify that $\bm{Y}_{1,8} < \bm{Y}_{2,8}$ holds. Hence, the algorithm identifies that the last column has been deleted. 

In Step~2, we compute the vector 
\[
\bm{y} = (2,1,0,1,0,6,5,6),
\] 
and insert it into the last column of $\bm{Y}$. After this insertion, the resulting array (still denoted by $\bm{Y}$) is an $8 \times 9$ array given by
\[
\bm{Y} =
\begin{bmatrix}
4 & 2 & 1 & 4 & 5 & 2 & 1 & 0 & {\color{red}2} \\
5 & 0 & 2 & 4 & 4 & 0 & 4 & 1 & {\color{red}1} \\
5 & 1 & 2 & 6 & 0 & 3 & 2 & 2 & {\color{red}0} \\
5 & 6 & 1 & 0 & 3 & 2 & 4 & 6 & {\color{red}1} \\
2 & 4 & 3 & 5 & 6 & 1 & 6 & 1 & {\color{red}0} \\
3 & 1 & 3 & 2 & 3 & 2 & 2 & 6 & {\color{red}6} \\
5 & 0 & 0 & 5 & 0 & 5 & 5 & 3 & {\color{red}5} \\
3 & 0 & 6 & 2 & 3 & 3 & 0 & 5 & {\color{red}6} \\
\end{bmatrix}.
\]

In Step~4, we extract the vector 
\[
\bm{V}' = (6,5,6,0,1,0,1,2).
\] 
It can be verified that $\bm{V}'$ is obtained from the reversal of the last column of the original codeword array 
\[
\bm{V} = (0,6,5,6,0,1,0,1,2)
\] 
by deleting the first element. By applying the decoding algorithm of the Differential VT codes to $\bm{V}'$, we find that the deleted symbol is the first element of $\bm{V}$. 
Since $\bm{V}$ is the reversal of the last column of the original codeword array $\bm{X}$, this first element in $\bm{V}$ is located in the last row of $\bm{X}$. 
Hence, we conclude that the last row of $\bm{X}$ was deleted, and we set $i^* = 9$ in Step~4.

Next, in Step~5, we obtain $n' = 9$ and compute 
\[
\bm{r} = (3,0,3,0,4,3,4,4,0).
\] 
By inserting $\bm{r}$ as the last row of $\bm{Y}$, we can verify that the resulting array coincides with the original codeword $\bm{X}$.  
In Step~7, since $n' = 9 \ne 8$, the algorithm directly returns the current $\bm{Y}$.  
Therefore, Algorithm~\ref{alg:decoding} successfully recovers the original codeword.

\section{One-Dimensional 1-RLL Differential VT Code with Suffix Constraint}\label{sec:1RLLDVT}

In order to encode the $(1,1)$-criss-cross deletion correcting code introduced in Section~\ref{sec:construction}, we first propose a new class of one-dimensional deletion correcting codes, termed the \emph{$q$-ary 1-RLL Differential VT code with suffix constraint}. Here, \emph{1-RLL} refers to the run-length-limited constraint of length one, meaning that adjacent symbols are required to be distinct (see~\cite{2024Differential_VT_codes}). These codes are specifically designed to correct one single deletion and identify its location. These codes serve as building blocks for the encoding procedure of \textbf{Construction~\ref{constrcution:ourcode}}. 

We first give a formal definition of the $q$-ary 1-RLL Differential VT code with suffix constraint. We then present the corresponding encoding, decoding, and data recovery algorithms.

\subsection{The $q$-ary 1-RLL Differential VT Code with Suffix Constraint}

In \textbf{Construction~\ref{constrcution:ourcode}}, the vectors $\bm{U}$ and $\bm{V}$ are not only required to be codewords of a one-dimensional Differential VT code, but also to satisfy two additional properties:  
(i) adjacent symbols must be distinct (the 1-RLL property); and  
(ii) the last few entries are fixed to prescribed values.  
We combine these requirements into the following definition of a new class of one-dimensional deletion correcting codes.

\begin{definition}[$q$-ary 1-RLL Differential VT Code with Suffix Constraint]
Let $q \ge 3$, $n \ge 2$, $m\ge1$ be integers, $a \in \mathbb{Z}_{q(n+m)}$ be a syndrome parameter, and $\bm{b} = (b_1, \ldots, b_m) \in \Sigma_q^m$ be a suffix sequence satisfying $b_i \ne b_{i+1}$ for all $1 \le i \le m-1$. The \emph{$q$-ary 1-RLL Differential VT code with suffix constraint}, denoted as
\[
\mathrm{1\text{-}RLL\_Diff\_VT}_a(n,m; \bm{b}),
\]
is the set of sequences $\bm{x} \in \Sigma_q^{n+m}$ that satisfy the following conditions:
\begin{enumerate}
  \item $\mathrm{Syn}(\mathrm{Diff}(\bm{x})) \equiv a \ppmod{q(n+m)}$;
  \item $x_i \ne x_{i+1}$ for all $1 \le i \le n+m-1$;
  \item $(x_{n+1}, \ldots, x_{n+m}) = (b_1, \ldots, b_m)$.
\end{enumerate}
\end{definition}

Clearly, in \textbf{Construction~\ref{constrcution:ourcode}}, the vectors $\bm{U}$ and $\bm{V}$ belong to this class of codes. In particular,  
\begin{align*}
\bm{U} \in \mathrm{1\text{-}RLL\_Diff\_VT}_0(n-2,2;(0,2)), \\
\bm{V} \in \mathrm{1\text{-}RLL\_Diff\_VT}_0(n-3,3;(0,1,2)).
\end{align*}

\subsection{Encoding Algorithm of Code $\mathrm{1\text{-}RLL\_Diff\_VT}_a(n,m; \bm{b})$}\label{subsec:1RLLDVT_encoding}

We now present an encoding procedure that maps a sequence of data symbols to a valid codeword in $\mathrm{1\text{-}RLL\_Diff\_VT}_a(n,m; \bm{b})$.

More precisely, for given parameters $n \ge 8$, $m \le 3$, $q \ge 3$, and a fixed suffix sequence $\bm{b} \in \Sigma_q^m$ with $b_i \ne b_{i+1}$ for all $1 \le i \le m-1$, let $t=\lfloor \log_{q-1} n \rfloor$, we aim to map the sequence of data symbols

$$
(f_1, f_2, \ldots, f_k) \in \mathbb{Z}_{q-1}^{k}, 
\quad \text{where } k = n - t - 4,
$$
into a valid codeword in the set $\mathrm{1\text{-}RLL\_Diff\_VT}_a(n,m; \bm{b})$.
Note that when $n \ge 8$ and $q \ge 3$, we have $n-t-4 \ge 1$, and hence the above sequence of data symbols is well defined.


Algorithm~\ref{alg:1RLLDVT_encoding} presents the encoding procedure for code $\mathrm{1\text{-}RLL\_Diff\_VT}_a(n,m; \bm{b})$. 
In general, both data symbols and coded symbols are taken from the same alphabet $\mathbb{Z}_q$. In the present encoding procedure, however, in order to satisfy the 1-RLL property, we restrict the data symbols to $\mathbb{Z}_{q-1}$. This guaranties that for each encoding step, there remains at least one available choice distinct from the previous symbol. Such a design simplifies the encoding procedure and ensures that the resulting codewords satisfy the 1-RLL property.
To facilitate understanding, we illustrate the procedure of Algorithm~\ref{alg:1RLLDVT_encoding} with a concrete example before proving its validity.


\begin{algorithm}
\caption{Encoding Algorithm for code $\mathrm{1\text{-}RLL\_Diff\_VT}_a(n,m; \bm{b})$}
\label{alg:1RLLDVT_encoding}
\begin{algorithmic}[1]
\Require A sequence of data symbols $(f_1, f_2, \ldots, f_{n - t - 4}) \in \mathbb{Z}_{q-1}^{\,n - t - 4}$, a suffix sequence $\bm{b} = (b_1, \ldots, b_m) \in \Sigma_q^m$ with $b_i \ne b_{i+1}$ for $1 \le i \le m-1$, parameters $n \ge 8$, $m \le 3$, $q \ge 3$, and a syndrome $a \in \mathbb{Z}_{q(n+m)}$.
\Ensure A codeword $\bm{x} \in \mathrm{1\text{-}RLL\_Diff\_VT}_a(n,m; \bm{b})$.

\State \textbf{Initialization:} $R_1 = \{ (q-1)^i \mid 0 \le i \le t \}$, $R_2 = \emptyset$, $\bm{y} = (0,0,\ldots,0) \in \Sigma_q^{n+m}$

\For{$i = n$ \textbf{downto} $1$}
  \If{$i \notin R_1$ \textbf{and} $|R_2| < 3$}
    \State $R_2 = R_2 \cup \{i\}$
  \EndIf
\EndFor

\State Sort $R_2$ in ascending order as $\{j_1, j_2, j_3\}$, and let $R = R_1 \cup R_2$, $K = \{1,2,\ldots,n\} \setminus R$, where $\setminus$ denotes the set difference operation

\For{$i = 1$ \textbf{to} $n-t-4$}
  \State Let $k_i$ be the $i$-th smallest element of $K$
  \State Set $y_{k_i} = f_i + 1$
\EndFor

\For{$i = n+1$ \textbf{to} $n+m-1$}
  \State $y_i = b_{i-n} - b_{i-n+1} \ppmod q$
\EndFor
\State $y_{n+m} = b_m$

\State Compute
\[
g_1 =a - \sum_{i=n+1}^{n+m} i y_i - \sum_{i \in K} i y_i - \sum_{i \in R} i  \ppmod {q(n+m)}
\]

\For{$i = 1$ \textbf{to} $3$}
  \State $e_i = \min(q-2, \lfloor g_i / j_i \rfloor)$
  \State $y_{j_i} = e_i + 1$
  \State $g_{i+1} = g_i - e_i \cdot j_i$
\EndFor

\State Express $g_4$ in base $(q-1)$:
\begin{equation}\label{eq:g4expansion}
g_4 = \sum_{i=0}^{t} h_i (q-1)^i, \quad 0 \le h_i \le q-2
\end{equation}

\For{$i = 0$ \textbf{to} $t$}
  \State $y_{(q-1)^i} = h_i + 1$
\EndFor

\State \Return $\bm{x} = \mathrm{Diff}^{-1}(\bm{y})$
\end{algorithmic}
\end{algorithm}

\begin{example}\label{example:onedim_encoding}
Consider the input parameters $n=7$, $q=7$, $m=2$, $\bm{b}=(0,2)$, $a=0$, $t=\lfloor \log_{6} 7\rfloor=1$, and the data symbols $(f_1=0, f_2=3)$. 
Note that the algorithm formally requires $n \ge 8$, but for the sake of illustration, we take $n=7$. 

In the initialization, we obtain $R_1 = \{1,6\}$. 
Steps~2--6 yield $R_2 = \{4,5,7\}$, and Step~7 gives $K = \{2,3\}$ and $R = \{1,4,5,6,7\}$. 
Steps~8--11 assign $y_2=f_1+1=1$, $y_3=f_2+1=4$ (these positions are used to store the data symbols), while Steps~12--15 assign $y_8=5$, $y_9=2$ (this ensures that the suffix of $\bm{x}$ matches the prescribed sequence $\bm{b}$). At this point, we obtain 
\[
\bm{y} = (0,1,4,0,0,0,0,5,2).
\] 

The subsequent steps are designed to guarantee that the final sequence belongs to the Differential VT code. In Step~16, we compute $g_1=31$. 
Iterating through Steps~17--21, we successively obtain 
$e_1=5$, $y_4=6$, $g_2=11$; 
$e_2=2$, $y_5=3$, $g_3=1$; 
$e_3=0$, $y_7=1$, $g_4=1$. 
Step~22 then gives $h_0=1$, $h_1=0$, and Steps~23--25 assign $y_1=2$, $y_6=1$. 
At this point, 
\[
\bm{y} = (2,1,4,6,3,1,1,5,2).
\]

Finally, Step~26 computes 
\[
\bm{x} = (4,2,1,4,5,2,1,0,2).
\] 
It can be verified that $\bm{x} \in \mathrm{1\text{-}RLL\_Diff\_VT}_0(7,2;(0,2))$.
\end{example}

We now turn to the formal analysis of Algorithm~\ref{alg:1RLLDVT_encoding}.
Lemma~\ref{lemma:1RLLDVT_welldesign} establishes the feasibility of its intermediate steps, 
and the correctness is proved in Theorem~\ref{theorem:1RLLDVT_encoding_correctness}.

\begin{lemma}
\label{lemma:1RLLDVT_welldesign}
For $n \ge 8$, $q \ge 3$, and $m \le 3$, Algorithm~\ref{alg:1RLLDVT_encoding} satisfies the following properties:
\begin{enumerate}
    \item The set $R_2$ obtained in Steps~2--6 has cardinality $|R_2|=3$;
    \item The index set $K$ in Step~7 satisfies $|K| = n-t-4$;
    \item The index $j_1$ obtained in Step~7 satisfies $j_1 \ge n-3$;
    \item The value $g_4$ in Step~22 satisfies 
    \[
    0 \le g_4 \le (q-1)^{\,t+1}-1,
    \]
    ensuring that the expansion in Eq.~\eqref{eq:g4expansion} is valid.
\end{enumerate}
\end{lemma}

\begin{IEEEproof}
We prove each property in turn:

\begin{enumerate}
    \item Since $n \ge 8$ and $q \ge 3$, we have $n-(t+1) > 3$, which guarantees that the set $R_2$ obtained in Steps~2--6 of Algorithm~\ref{alg:1RLLDVT_encoding} satisfies $|R_2| = 3$.

    \item 

From the definitions of $R_1$ and $t$, we have $|R_1| = t+1$, and all elements of $R_1$ belong to the set $\{1,2,\ldots,n\}$. Moreover, it has been established that $|R_2| = 3$, and from Steps~2--6 of the algorithm, $R_1$ and $R_2$ are disjoint. Therefore, in Step~7, the set $K = \{1,2,\ldots,n\} \setminus (R_1 \cup R_2)$ satisfies

$$
|K| = n - |R_1| - |R_2| = n - t - 4.
$$

    \item We next prove by contradiction that $j_1 \ge n-3$ in Step~7. Suppose, to the contrary, that $j_1 < n-3$. Then, among the four integers $n-3, n-2, n-1, n$, at least two must be of the form $(q-1)^i$ for some $0 \le i \le t$. That is, there exist integers $i_1, i_2$ with $0 \le i_1 < i_2 \le t$ such that $(q-1)^{i_1} \ge n-3$ and $(q-1)^{i_2} \le n$, implying
\[
(q-1)^{i_2} - (q-1)^{i_1} \le 3.
\]

Since $i_2 \ge i_1 + 1$, we have $(q-1)^{i_1+1} - (q-1)^{i_1} = (q-1)^{i_1}(q-2) \le 3$. For $q \ge 6$, we have $q-2 \ge 4$, so this inequality cannot hold. For $3 \le q \le 5$, since $(q-1)^{i_1} \ge n-3 \ge 5$, it follows that $i_1 \ge 2$, which leads to $(q-1)^2(q-2) \le 3$. Substituting $q=3,4,5$ shows that this is impossible, yielding a contradiction. Hence, $j_1 \ge n-3$.

\item Next, we show that the value $g_4$ in Step~22 satisfies $0 \le g_4 \le (q-1)^{t+1}-1$. Step~16 defines $g_1$ such that $0 \le g_1 \le q(n+m)-1$. By Steps~18--20, we have $$g_2 = g_1 - \min(q-2, \lfloor g_1 / j_1 \rfloor) \cdot j_1.$$

Consider the two cases: 
\begin{itemize}
\item If $g_1 \ge (q-2)j_1$, then $g_2 = g_1 - (q-2)j_1$, thus $0 \le g_2 \le g_1 - (q-2)j_1$.
\item If $g_1 < (q-2) j_1$, then $g_2 = g_1 - \lfloor g_1/j_1 \rfloor j_1$, which implies $0 \le g_2 \le j_1-1$.
\end{itemize}

In general, we have
\[
0 \le g_2 \le \max(g_1 - (q-2) j_1, j_1-1),
\]
and since $j_1 - 1 \le n-1$, it follows that $$0 \le g_2 \le \max(g_1 - (q-2) j_1, n-1).$$ Similarly, we obtain
\begin{equation*}
\begin{aligned}
0 &\le g_3 \le \max\bigl(g_1 - (q-2) j_1 - (q-2) j_2, \, n-1\bigr),\\
0 &\le g_4 \le 
\max\Bigl(
    g_1 - (q-2) j_1 - (q-2) j_2 - (q-2) j_3, n-1
\Bigr).
\end{aligned}
\end{equation*}

Since $n \le (q-1)^{t+1}$, to ensure $0 \le g_4 \le (q-1)^{t+1}-1$, it suffices to show that
\[
g_1 - (q-2) j_1 - (q-2) j_2 - (q-2) j_3 \le (q-1)^{t+1}-1.
\]
Since $j_1 \ge n-3$, we have $j_2 \ge n-2$ and $j_3 \ge n-1$. Thus,
\begin{align*}
g_1& - (q-2) j_1 - (q-2) j_2 - (q-2) j_3\\
\le\,& q(n+m) - 1 - (q-2)(3n-6) \\
\le\,& q(9-2n) + 6n - 13,
\end{align*}
where we have used $m \le 3$. Hence, it remains to show
\[
q(9-2n) + 6n - 13 \le (q-1)^{t+1} - 1.
\]
When $q=3$ and $8 \le n \le 14$, the inequality can be verified directly by enumeration. In all other cases, since $(q-1)^{t+1}-1 \ge n-1$, it suffices to show
\[
q(9-2n) + 6n - 13 \le n-1,
\]
which is equivalent to $n \ge \frac{9q-12}{2q-5}$. It is easy to verify that this holds for $q=3, n \ge 15$ and for $q \ge 4, n \ge 8$. Therefore, $g_4$ satisfies $0 \le g_4 \le (q-1)^{t+1}-1$, ensuring that the expansion in Eq.~\eqref{eq:g4expansion} is valid.

\end{enumerate}

Combining these properties, we see that every operation in Algorithm~\ref{alg:1RLLDVT_encoding} can be performed as intended, ensuring the algorithm is well-defined.
\end{IEEEproof}

Based on the properties established in Lemma~\ref{lemma:1RLLDVT_welldesign}, we can now prove the overall correctness of the algorithm.

\begin{theorem}
\label{theorem:1RLLDVT_encoding_correctness}
For $n \ge 8$, $q \ge 3$, and $m \le 3$, the output $\bm{x}$ of Algorithm~\ref{alg:1RLLDVT_encoding} satisfies
\[
\bm{x} \in \mathrm{1\text{-}RLL\_Diff\_VT}_a(n,m; \bm{b}).
\]
\end{theorem}

\begin{proof}
We verify that the output $\bm{x}$ satisfies the following three properties:

\begin{enumerate}
    \item \textbf{Suffix constraint.}  
    We first verify that 
    \[
        (x_{n+1}, \ldots, x_{n+m}) = (b_1, \ldots, b_m).
    \]
    We proceed by backward induction.  
    From Step~15 of the algorithm, $y_{n+m} = b_m$, and by Step~26, 
    $\bm{x} = \mathrm{Diff}^{-1}(\bm{y})$, which gives $x_{n+m} = b_m$.  

    Suppose that $x_{n+g} = b_g$ holds for some $2 \le g \le m$.  
    From Step~13, we have 
    \[
        y_{n+g-1} = b_{g-1} - b_g \ppmod q.
    \]
    On the other hand, since $\mathrm{Diff}(\bm{x}) = \bm{y}$, it follows that
    \[
        y_{n+g-1} = x_{n+g-1} - x_{n+g} \ppmod q.
    \]
    Therefore,
    \[
        b_{g-1} - b_g \equiv x_{n+g-1} - x_{n+g} \ppmod{q}.
    \]
    By the induction hypothesis $x_{n+g} = b_g$, which implies
    \[
        x_{n+g-1} \equiv b_{g-1} \ppmod{q}.
    \]
    Since both $x_{n+g-1}, b_{g-1} \in \Sigma_q$, we obtain $x_{n+g-1} = b_{g-1}$.  
    By induction, $(x_{n+1}, \ldots, x_{n+m}) = (b_1, \ldots, b_m)$.

    \item \textbf{1-RLL property.}  
    Next, we show that $x_i \ne x_{i+1}$ for all $1 \le i \le n+m-1$.  
    Since $b_i \ne b_{i+1}$ for $i=1,2,\ldots,m-1$ and 
    $(x_{n+1}, \ldots, x_{n+m}) = (b_1, \ldots, b_m)$, we have
    $x_i \ne x_{i+1}$ for $i = n+1,\ldots,n+m-1$.  

    From the assignments in Steps~10, 19, and 24 of the algorithm, we have
    $1 \le y_i \le q-1$ for all $i=1,2,\ldots,n$.  
    Since $y_i = x_i - x_{i+1} \ppmod q$, it follows that 
    $x_i \ne x_{i+1}$ for $i = 1,2,\ldots,n$.  
    Therefore, $x_i \ne x_{i+1}$ holds for all $1 \le i \le n+m-1$.

    \item \textbf{Syndrome constraint.}  
    Finally, we verify that
    \[
        \mathrm{Syn}(\mathrm{Diff}(\bm{x})) \equiv a \ppmod{q(n+m)}.
    \]
    Equivalently, we need to show
    \[
        \sum_{i=1}^{n+m} i y_i \equiv a \ppmod{q(n+m)}.
    \]
    From the definition of $g_1$ in Step~16 and the fact that $R=R_1 \cup R_2$, we obtain
    \begin{align*}
    &a - \sum_{i=1}^{n+m} i y_i \\[0.5ex]
    \equiv\,& a - \sum_{i \in R} i (y_i - 1) - \sum_{i \in R} i 
              - \sum_{i \in K} i y_i - \sum_{i=n+1}^{n+m} i y_i \\[0.5ex]
    \equiv\,& g_1 - \sum_{i \in R_1} i (y_i - 1) - \sum_{i \in R_2} i (y_i - 1) \\[0.5ex]
    \equiv\,& g_1 - \sum_{i=0}^t (q-1)^i \bigl(y_{(q-1)^i} - 1 \bigr) 
              - \sum_{i=1}^3 j_i (y_{j_i} - 1) \\
    &\ppmod{q(n+m)}.
\end{align*}

    From Steps~24 and 22, we have
    \[
        \sum_{i=0}^t (q-1)^i \bigl(y_{(q-1)^i} - 1 \bigr)
        = \sum_{i=0}^t h_i (q-1)^i = g_4.
    \]
    Moreover, from Steps~17--21,
    \[
        g_1 - \sum_{i=1}^3 j_i (y_{j_i} - 1) = g_4.
    \]
    Combining the above, we obtain
    \[
        a - \sum_{i=1}^{n+m} i y_i \equiv 0 \ppmod{q(n+m)}.
    \]
    Hence, $\mathrm{Syn}(\mathrm{Diff}(\bm{x})) \equiv a \ppmod{q(n+m)}$.
\end{enumerate}

Therefore, the output $\bm{x}$ of the algorithm satisfies 
$\bm{x} \in \mathrm{1\text{-}RLL\_Diff\_VT}_a(n,m; \bm{b})$, which completes the proof.

\end{proof}

\subsection{Decoding and Data Recovery Algorithms of Code $\mathrm{1\text{-}RLL\_Diff\_VT}_a(n,m; \bm{b})$}
In this subsection, we present the decoding and data recovery algorithms of code $\mathrm{1\text{-}RLL\_Diff\_VT}_a(n,m; \bm{b})$. 

Since this code is itself a Differential VT code with the additional property that adjacent symbols are distinct, the same decoding procedure described in Proposition~\ref{prop:unique_deletion_position} can be applied to uniquely reconstruct the original codeword and identify the location of a single deletion. 
After the original codeword $\bm{x}$ is recovered, the data recovery algorithm (which is presented in Algorithm~\ref{alg:1RLLDVT_info_recovery}) is applied to obtain the original data symbols, whose correctness is straightforward.

\begin{algorithm}
\caption{Data Recovery Algorithm for Code $\mathrm{1\text{-}RLL\_Diff\_VT}_a(n,m; \bm{b})$}
\label{alg:1RLLDVT_info_recovery}
\begin{algorithmic}[1]
\Require The recovered codeword $\bm{x} \in \mathrm{1\text{-}RLL\_Diff\_VT}_a(n,m; \bm{b})$
\Ensure Original sequence of data symbols $(f_1, f_2, \ldots, f_{n-t-4}) \in \mathbb{Z}_{q-1}^{n-t-4}$
\State Compute $\bm{y} = \mathrm{Diff}(\bm{x})$
\State Determine the index set $K$ according to Steps~1–7 of Algorithm~\ref{alg:1RLLDVT_encoding}
\State Extract the subsequence $(y_{k_1}, y_{k_2}, \ldots, y_{k_{n-t-4}})$ from $\bm{y}$ with indices in $K$
\State Subtract one from each element, i.e., $f_i = y_{k_i}-1$, to obtain $(f_1, f_2, \ldots, f_{n-t-4})$
\State \Return $(f_1, f_2, \ldots, f_{n-t-4})$
\end{algorithmic}
\end{algorithm}

It is clear that the complexities of the encoding algorithm (Algorithm~\ref{alg:1RLLDVT_encoding}), the decoding algorithm (refer to \cite{2024Differential_VT_codes}),  and the data recovery algorithms (Algorithm~\ref{alg:1RLLDVT_info_recovery}) are all $\mathcal{O}(n)$.

\section{Encoding and Data Recovery Algorithms}

In this section, building upon code $\mathrm{1\text{-}RLL\_Diff\_VT}_a(n,m; \bm{b})$ introduced in Section~\ref{sec:1RLLDVT}, 
we present the encoding and data recovery algorithms of the $(1,1)$-criss-cross deletion correcting code \textbf{Construction~\ref{constrcution:ourcode}} 
(note that the decoding algorithm has already been presented in Section~\ref{subsec:2Ddecoding}).

\subsection{Encoding Algorithm}


Let 
\begin{align*}
k_1 &= n-6-\lfloor \log_{q-1}(n-2) \rfloor,\\ 
k_2 &= n-7-\lfloor \log_{q-1}(n-3) \rfloor,\\ 
k_3 &= \lfloor (k_1+k_2)\log_q(q-1) \rfloor.
\end{align*}

Given the parameters $n \geq 11$ and $q \geq 3$, we next present an explicit encoding algorithm for \textbf{Construction~\ref{constrcution:ourcode}}, which achieves an encoder redundancy of 
\[
r_{\mathrm{ENC}} = 4n-2-k_3.
\]
Specifically, given a sequence of data symbols 
\[
(f_1,f_2,\ldots,f_{n^2-4n+2+k_3}) \in \mathbb{Z}_q^{\,n^2-4n+2+k_3},
\] 
our goal is to encode it into a codeword $\bm{X} \in \Sigma_q^{n \times n}$ such that $\bm{X} \in \mathcal{C}$, where $\mathcal{C}$ denotes the $(1,1)$-criss-cross deletion correcting code defined in \textbf{Construction~\ref{constrcution:ourcode}}.    

The key idea of our encoding algorithm is to map the first $k_3$ data symbols over $\mathbb{Z}_q$ into $k_1+k_2$ symbols over $\mathbb{Z}_{q-1}$. Among these, $k_1$ symbols are used in the encoding of $\bm{U}$ and the remaining $k_2$ symbols are used in the encoding of $\bm{V}$ in \textbf{Construction~\ref{constrcution:ourcode}}. The subsequent $n^2-4n+2$ data symbols are directly placed into the central part of the array as shown in Fig.~\ref{fig:1}. Finally, the vectors $\bm{p_c}$ and $\bm{p_r}$ of \textbf{Construction~\ref{constrcution:ourcode}} are computed to complete the encoding. The detailed encoding procedure is summarized in Algorithm~\ref{alg:ourencoding}. Note that in Step~3, the expansion contains only $k_1+k_2$ terms because $0 \le h \le q^{k_3}-1$ and $k_3 = \lfloor (k_1+k_2)\log_q(q-1) \rfloor$. The correctness of Algorithm~\ref{alg:ourencoding} is straightforward and the proof is omitted.

\begin{algorithm}[!t]
\caption{Encoding Algorithm for \textbf{Construction~\ref{constrcution:ourcode}}}
\label{alg:ourencoding}
\begin{algorithmic}[1]
\Require Parameters $n \geq 11$, $q \geq 3$; a sequence of data symbols $(f_1,f_2,\ldots,f_{n^2-4n+2+k_3}) \in \mathbb{Z}_q^{n^2-4n+2+k_3}$
\Ensure Codeword $\bm{X} \in \mathcal{C}$, where $\mathcal{C}$ is the $(1,1)$-criss-cross deletion correcting code from \textbf{Construction~\ref{constrcution:ourcode}}
\State Initialize $\bm{X} \in \Sigma_q^{n\times n}$ as the all-zero array
\State Compute $h = \sum\limits_{i=1}^{k_3} f_i q^{i-1}$
\State Expand $h$ in base $(q-1)$: $h = \sum\limits_{i=1}^{k_1+k_2} a_i (q-1)^{i-1}$, where $0 \le a_i \le q-2$
\State Set $\bm{c} = (a_1,a_2,\ldots,a_{k_1})$; encode $\bm{c}$ into $\bm{U} \in \mathrm{1\text{-}RLL\_Diff\_VT}_0(n-2,2;(0,2))$ using Algorithm~\ref{alg:1RLLDVT_encoding}, and place $\bm{U}$ in the first row of $\bm{X}$
\State Set $\bm{d} = (a_{k_1+1},\ldots,a_{k_1+k_2})$; encode $\bm{d}$ into $\bm{V} \in \mathrm{1\text{-}RLL\_Diff\_VT}_0(n-3,3;(0,1,2))$ using Algorithm~\ref{alg:1RLLDVT_encoding}, and place the reversed $\bm{V}$ in the last column of $\bm{X}$
\State Assign $\bm{X}_{2,n-1} = 1$, $\bm{X}_{3,n-1} = 2$
\For{$j = 2$ \textbf{to} $n-2$}
    \State $\bm{X}_{2,j} = f_{k_3+j-1}$
    \State $\bm{X}_{3,j} = f_{k_3+j+n-4}$
\EndFor
\For{$i = 4$ \textbf{to} $n-1$}
    \For{$j = 2$ \textbf{to} $n-1$}
        \State $\bm{X}_{i,j} = f_{k_3 + (n-2)i + j - 2n + 1}$
    \EndFor
\EndFor
\For{$i = 2$ \textbf{to} $n$}
    \State $\bm{X}_{i,1} = -\sum\limits_{j=2}^n \bm{X}_{i,j} \ppmod q$
\EndFor
\For{$j = 2$ \textbf{to} $n-1$}
    \State $\bm{X}_{n,j} = -\sum\limits_{i=1}^{n-1} \bm{X}_{i,j}\ppmod q$
\EndFor
\State \Return $\bm{X}$
\end{algorithmic}
\end{algorithm}

To illustrate the encoding procedure, we now detail how the codeword $\bm{X}$ in Example~\ref{example:encoding} is obtained using Algorithm~\ref{alg:ourencoding}.

Consider the case $n=9$, $q=7$, $k_1=2$, $k_2=1$, $k_3=2$, and the sequence of data symbols
\begin{align*}
(f_1&, \ldots, f_{49}) \\
= (&4,2,0,2,4,4,0,4,1,2,6,0,3,2,6,1,0,3,2,4,\\
& 6,4,3,5,6,1,6,1,1,3,2,3,2,2,6,0,0,5,0,5,\\
&5,3,0,6,2,3,3,0,5).
\end{align*}
Note that although the algorithm requires $n \geq 11$, this condition is relaxed here for illustrative purposes. 

In Step~2, we compute $h=4+2\times 7=18$. Step~3 expands $18=3\times 6$, which gives $a_1=0$, $a_2=3$, and $a_3=0$. In Step~4, we obtain $\bm{c}=(0,3)$. The encoding of $\bm{c}$ by Algorithm~\ref{alg:1RLLDVT_encoding}, as demonstrated in Example~\ref{example:onedim_encoding}, yields $\bm{U}=(4,2,1,4,5,2,1,0,2)$. In Step~5, $\bm{d}=(0)$, and the corresponding encoding produces $\bm{V}=(0,6,5,6,0,1,0,1,2)$. 
Placing $\bm{U}$, $\bm{V}$, and two fixed symbols from Step~6 into $\bm{X}$ gives
\[
\bm{X} =
\begin{bmatrix}
\color{red}{4} & \color{red}{2} & \color{red}{1} & \color{red}{4} & \color{red}{5} & \color{red}{2} & \color{red}{1} & \color{red}{0} & \color{red}{2} \\
0 & 0 & 0 & 0 & 0 & 0 & 0 & \color{red}{1} & \color{red}{1} \\
0 & 0 & 0 & 0 & 0 & 0 & 0 & \color{red}{2} & \color{red}{0} \\
0 & 0 & 0 & 0 & 0 & 0 & 0 & 0 & \color{red}{1} \\
0 & 0 & 0 & 0 & 0 & 0 & 0 & 0 & \color{red}{0} \\
0 & 0 & 0 & 0 & 0 & 0 & 0 & 0 & \color{red}{6} \\
0 & 0 & 0 & 0 & 0 & 0 & 0 & 0 & \color{red}{5} \\
0 & 0 & 0 & 0 & 0 & 0 & 0 & 0 & \color{red}{6} \\
0 & 0 & 0 & 0 & 0 & 0 & 0 & 0 & \color{red}{0}
\end{bmatrix}.
\]

Steps~7--15 successively place the remaining $47$ information symbols 
$(f_3,f_4,\ldots,f_{49})$ into the interior of $\bm{X}$. 
The entries that are filled in this step are highlighted in red in the following array:
\[
\begin{bmatrix}
4 & 2 & 1 & 4 & 5 & 2 & 1 & 0 & 2 \\
0 & \color{red}{0} & \color{red}{2} & \color{red}{4} & \color{red}{4} & \color{red}{0} & \color{red}{4} & 1 & 1 \\
0 & \color{red}{1} & \color{red}{2} & \color{red}{6} & \color{red}{0} & \color{red}{3} & \color{red}{2} & 2 & 0 \\
0 & \color{red}{6} & \color{red}{1} & \color{red}{0} & \color{red}{3} & \color{red}{2} & \color{red}{4} & \color{red}{6} & 1 \\
0 & \color{red}{4} & \color{red}{3} & \color{red}{5} & \color{red}{6} & \color{red}{1} & \color{red}{6} & \color{red}{1} & 0 \\
0 & \color{red}{1} & \color{red}{3} & \color{red}{2} & \color{red}{3} & \color{red}{2} & \color{red}{2} & \color{red}{6} & 6 \\
0 & \color{red}{0} & \color{red}{0} & \color{red}{5} & \color{red}{0} & \color{red}{5} & \color{red}{5} & \color{red}{3} & 5 \\
0 & \color{red}{0} & \color{red}{6} & \color{red}{2} & \color{red}{3} & \color{red}{3} & \color{red}{0} & \color{red}{5} & 6 \\
0 & 0 & 0 & 0 & 0 & 0 & 0 & 0 & 0
\end{bmatrix}.
\]

Finally, Steps~16--21 compute the vectors $\bm{p_c}$ and $\bm{p_r}$ in \textbf{Construction~\ref{constrcution:ourcode}}, updating $\bm{X}$ to
\[\bm{X}=
\begin{bmatrix}
4 & 2 & 1 & 4 & 5 & 2 & 1 & 0 & 2 \\
\color{red}{5} & 0 & 2 & 4 & 4 & 0 & 4 & 1 & 1 \\
\color{red}{5} & 1 & 2 & 6 & 0 & 3 & 2 & 2 & 0 \\
\color{red}{5} & 6 & 1 & 0 & 3 & 2 & 4 & 6 & 1 \\
\color{red}{2} & 4 & 3 & 5 & 6 & 1 & 6 & 1 & 0 \\
\color{red}{3} & 1 & 3 & 2 & 3 & 2 & 2 & 6 & 6 \\
\color{red}{5} & 0 & 0 & 5 & 0 & 5 & 5 & 3 & 5 \\
\color{red}{3} & 0 & 6 & 2 & 3 & 3 & 0 & 5 & 6 \\
\color{red}{3} & \color{red}{0} & \color{red}{3} & \color{red}{0} & \color{red}{4} & \color{red}{3} & \color{red}{4} & \color{red}{4} & 0
\end{bmatrix}.
\]

It can be observed that the resulting array $\bm{X}$ coincides with Example~\ref{example:encoding}, and thus $\bm{X}\in\mathcal{C}$ as required.


\subsection{Data Recovery Algorithm}

When a $(1,1)$-criss-cross deletion occurs, the decoding algorithm in Section~\ref{subsec:2Ddecoding} recovers the original codeword $\bm{X}$. In this subsection, we further demonstrate that the original sequence of data symbols $(f_1,f_2,\ldots,f_{n^2-4n+2+k_3})$ can also be uniquely retrieved. 

Algorithm~\ref{alg:information_retrieval} formally presents the data recovery algorithm of \textbf{Construction~\ref{constrcution:ourcode}}, which essentially performs the inverse operation of Algorithm~\ref{alg:ourencoding}, and its correctness can be verified directly.

\begin{algorithm}
\caption{Data Recovery Algorithm for \textbf{Construction~\ref{constrcution:ourcode}}}
\label{alg:information_retrieval}
\begin{algorithmic}[1]
\Require Decoded codeword $\bm{X} \in \mathcal{C}$
\Ensure Original sequence of data symbols $(f_1,f_2,\ldots,f_{n^2-4n+2+k_3}) \in \mathbb{Z}_{q}^{n^2-4n+2+k_3}$
\State Extract the first row of $\bm{X}$ as $\bm{U}$, and apply Algorithm~\ref{alg:1RLLDVT_info_recovery} to recover the sequence of data symbols $(a_1,a_2,\ldots,a_{k_1})$
\State Extract the last column of $\bm{X}$, reverse its order, denote it as $\bm{V}$, and apply Algorithm~\ref{alg:1RLLDVT_info_recovery} to recover the sequence of data symbols $(a_{k_1+1},a_{k_1+2},\ldots,a_{k_1+k_2})$
\State Compute $h = \sum\limits_{i=1}^{k_1+k_2} a_i (q-1)^{i-1}$
\State Express $h$ in base-$q$, i.e., $h = \sum\limits_{i=1}^{k_3} f_i q^{i-1}$, where $0 \le f_i \le q-1$
\For{$j = 2$ \textbf{to} $n-2$}
    \State $f_{k_3+j-1} = \bm{X}_{2,j}$
    \State $f_{k_3+j+n-4} = \bm{X}_{3,j}$
\EndFor
\For{$i = 4$ \textbf{to} $n-1$}
    \For{$j = 2$ \textbf{to} $n-1$}
        \State $f_{k_3 + (n-2)i + j - 2n + 1} = \bm{X}_{i,j}$
    \EndFor
\EndFor
\State \Return $(f_1,f_2,\ldots,f_{n^2-4n+2+k_3}) \in \mathbb{Z}_{q}^{n^2-4n+2+k_3}$
\end{algorithmic}
\end{algorithm}

It is easy to verify that both the encoding algorithm and the data recovery algorithm have computational complexity $\mathcal{O}(n^2)$.



\section{Redundancy Analysis}
In this section, we analyze the encoder redundancy of \textbf{\textbf{Construction~\ref{constrcution:ourcode}}} and demonstrate that, under suitable parameter regimes, it  matches the lower bound in Theorem~\ref{theorem:theoretic_bound}, differing only by a constant \(\mathcal{O}(1)\).

When $n \ge 11$ and $q \ge 3$, the encoding algorithm of \textbf{Construction~\ref{constrcution:ourcode}} (Algorithm~\ref{alg:ourencoding}) maps $n^2 - 4n + 2 + k_3$ data symbols over $\mathbb{Z}_q$ into $n^2$ code symbols over $\mathbb{Z}_q$. Therefore, the encoder redundancy of \textbf{Construction~\ref{constrcution:ourcode}} is given by
\begin{equation}\label{eq:redundancy}
    r_{\mathrm{ENC}} = 4n - 2 - k_3.
\end{equation}

The following theorem provides an explicit upper bound on the encoder redundancy 
of \textbf{Construction~\ref{constrcution:ourcode}}.

\begin{theorem}
For $n \ge 11$ and $q \ge 3$, the encoder redundancy of \textbf{Construction~\ref{constrcution:ourcode}} satisfies
    \[
        r_{\mathrm{ENC}} \le 2n + 2\log_q n + (2n-13)\log_q\frac{q}{q-1} + 12.
    \]
    Moreover, if $q = \Omega(n)$, i.e., there exists a constant $c>0$ such that $q \ge cn$, then $$r_{\mathrm{ENC}} = 2n + 2\log_q n + \mathcal{O}(1).$$ 

\end{theorem}

\begin{proof}
    By the definitions of $k_1$, $k_2$, and $k_3$, we have
    \begin{align*}
        k_3
        &= \lfloor (k_1 + k_2) \log_q (q-1) \rfloor \\
        &\ge \Bigl( 2n - 13 - \log_{q-1}(n-2) - \log_{q-1}(n-3) \Bigr) \log_q (q-1) - 1 \\
        &\ge (2n-13)\log_q(q-1) - 2 \log_q n - 1 \\
        &= 2n - 14 + (2n-13) \log_q \frac{q-1}{q} - 2 \log_q n.
    \end{align*}
    Therefore, combining with Eq.~\eqref{eq:redundancy}, the encoder redundancy satisfies
    \[
        r_{\mathrm{ENC}} \le 2n + 2\log_q n + (2n-13) \log_q \frac{q}{q-1} + 12.
    \]

   When $q = \Omega(n)$, i.e., there exists a constant $c>0$ such that $q \ge cn$.  
    To bound the third term, observe that
    \[
        (2n-13)\log_q \frac{q}{q-1}
        = \frac{(2n-13)\,\ln\!\left(1+\tfrac{1}{q-1}\right)}{\ln q},
    \]
    where $\ln$ denotes the natural logarithm (base $e$).
    Using the inequality $\ln(1+x) \le x$ for $x>-1$, we obtain
    \[
        (2n-13)\log_q \frac{q}{q-1}
        \le \frac{2n-13}{(q-1)\ln q}.
    \]

    Since $q \ge cn$ and $cn > 1$ for sufficiently large $n$, the above is at most
    \[
        \frac{2n-13}{(cn-1)\ln(cn)},
    \]
     which for sufficiently large $n$ is bounded by
        \[
            \frac{4}{c} = \mathcal{O}(1).
        \] Substituting back, we obtain
    \[
        r_{\mathrm{ENC}} = 2n + 2\log_q n + \mathcal{O}(1).
    \]


\end{proof}

Since the code redundancy cannot exceed the encoder redundancy, the fact that the encoder redundancy attains the lower bound within an $\mathcal{O}(1)$ gap implies that the code redundancy is also optimal within $\mathcal{O}(1)$ gap.

\section{Conclusion}

In this paper, 
we propose a new construction of $q$-ary $(1,1)$-criss-cross deletion correcting code and provide 
explicit encoding, decoding, and data recovery algorithms. 
The proposed algorithms achieve a computational complexity of $\mathcal{O}(n^2)$. 
Moreover, we prove that under certain parameter settings, both the code redundancy and the 
encoder redundancy of the constructed criss-cross codes attain the lower bound within an $\mathcal{O}(1)$ gap. 
To the best of our knowledge, this is the first code construction that simultaneously attains the optimal code redundancy and encoder redundancy within an $\mathcal{O}(1)$ gap, while providing a explicit encoding algorithm.

\appendices
\section{Remarks on the Construction of~\cite{2025arxiv_crisscross3}}\label{appendix:2025problem}

In Section IV of~\cite{2025arxiv_crisscross3}, Construction 1 imposes Row/Column Sum Constraints such that the sum of the entries in the $k$-th row equals $b_k$, and the sum of the entries in the $k$-th column equals $a_k$, for $k=1,2,\ldots,n$. The first step of the decoding algorithm in~\cite{2025arxiv_crisscross3} attempts to recover the deleted symbols using these row and column sums. While this approach is natural, certain cases may lead to ambiguities. 

To illustrate a potential difficulty, consider the case $n=2$ and $q=4$, and the arrays
\[
A = \begin{bmatrix}
    1 & 1 \\
    1 & 2
\end{bmatrix}, \quad 
B = \begin{bmatrix}
    2 & 0 \\
    0 & 3
\end{bmatrix}.
\]
It is easy to verify that both arrays $A$ and $B$ have identical row sums and column sums. Now, if we delete the first row and first column of $A$ and the second row and second column of $B$, the resulting array in both cases is 
\[
C = \begin{bmatrix} 2 \end{bmatrix}.
\]
Therefore, it is impossible to determine the deleted symbols from $C$, since the indices of the deleted rows and columns in the construction of~\cite{2025arxiv_crisscross3} are not known a priori. Consequently, the decoding algorithm cannot be directly applicable in such cases.

One possible way to avoid this ambiguity would be to require all $a_k$’s and $b_k$’s to be identical. However, this modification is not compatible with the scheme of~\cite{2025arxiv_crisscross3}, where $a_k$ and $b_k$, $k=1,2,\ldots,n$, are fixed values chosen to facilitate subsequent arguments.




\section{Remarks on the Construction of~\cite{2021TIT_crisscross1}}\label{appendix:2021TIT_problem}

In Section VI-B of~\cite{2021TIT_crisscross1}, the deletion correcting decoder presented in the proof of Theorem 13 is intended to recover deleted rows and columns for the proposed code construction. In Case 1(b), it is stated that after correcting the column errors, the index of the deleted row can be determined using $\mathcal{VT}_{n-1,n}(c,d)$. It should be noted, however, that $\mathcal{VT}_{n-1,n}(c,d)$ does not include the last row, 
which makes it impossible to distinguish between the deletion of the penultimate row and the deletion of the last row in certain scenarios.

To illustrate this  limitation, consider the case $n=16$ and two $16 \times 16$ arrays $A$ and $B$:\setcounter{MaxMatrixCols}{20}
\[
\scalebox{1.0}{$
\setcounter{MaxMatrixCols}{20}
A=
\left[
\begin{array}{cccccccccccc|cccc}
0 & 0 & 0 & 0 & 0 & 0 & 0 & 0 & 0 & 0 & 0 & 0 & 1 & 1 & 0 & 0 \\
0 & 1 & 0 & 1 & 0 & 1 & 0& 1 & 0 & 1 & 0 & 1 & 1 & 0 & 0 & 1 \\
0 & 0 & 0 & 0 & 0 & 0 & 0 & 0 & 0 & 0 & 0 & 0 & 1 & 1 & 0 & 0 \\
0 & 1 & 0 & 1 & 0 & 1 & 0& 1 & 0 & 1 & 0 & 1 & 1 & 0 & 0 & 1 \\
\hline
0 & 0 & 0 & 0 & 0 & 0 & 0 & 0 & 0 & 0 & 0 & 0 & 0 & 0 & 0 & 0 \\
0 & 0 & 0 & 0 & 0 & 0 & 0 & 0 & 0 & 0 & 0 & 0 & 0 & 1 & 0 & 1 \\
0 & 0 & 0 & 0 & 0 & 0 & 0 & 0 & 0 & 0 & 0 & 0 & 0 & 0 & 0 & 0 \\
1 & 0 & 0 & 0 & 0 & 0 & 0 & 0 & 0 & 0 & 0 & 0 & 1 & 1 & 0 & 1 \\
0 & 0 & 0 & 0 & 0 & 0 & 0 & 0 & 0 & 0 & 0 & 0 & 0 & 0 & 0 & 0 \\
1 & 0 & 0 & 0 & 0 & 0 & 0 & 0 & 0 & 0 & 0 & 0 & 1 & 0 & 0 & 0 \\
0 & 0 & 0 & 0 & 0 & 0 & 0 & 0 & 0 & 0 & 0 & 0 & 0 & 0 & 1 & 1 \\
1 & 0 & 0 & 0 & 0 & 0 & 0 & 0 & 0 & 0 & 0 & 0 & 0 & 1 & 1 &1 \\
1 & 0 & 0 & 0 & 0 & 0 & 0 & 0 & 0 & 0 & 0 & 0 & 1 & 0 & 0 & 0 \\
0 & 1 & 1 & 0 & 0 & 0 & 0 & 0 & 0 & 0 & 0 & 0 & 1 & 1 & 0 & 0 \\
0 & \textcolor{red}{1} & \textcolor{red}{1} & 0 & 0 & 0 & 0 & 0 & 0 & 0 & 0 & 0 & 0 & 0 & 0 & 0 \\
\hline
0 & \textcolor{red}{0} & \textcolor{red}{0} & 0 & 0 & 0 & 0 & 0 & 0 & 0 & 0 & 0 & 0 & 0 & 0 & 0 \\
\end{array}
\right],
$}
\]
\[
\scalebox{1.0}{$
\setcounter{MaxMatrixCols}{20}
B=
\left[
\begin{array}{cccccccccccc|cccc}
0 & 0 & 0 & 0 & 0 & 0 & 0 & 0 & 0 & 0 & 0 & 0 & 1 & 1 & 0 & 0 \\
0 & 1 & 0 & 1 & 0 & 1 & 0& 1 & 0 & 1 & 0 & 1 & 1 & 0 & 0 & 1 \\
0 & 0 & 0 & 0 & 0 & 0 & 0 & 0 & 0 & 0 & 0 & 0 & 1 & 1 & 0 & 0 \\
0 & 1 & 0 & 1 & 0 & 1 & 0& 1 & 0 & 1 & 0 & 1 & 1 & 0 & 0 & 1 \\
\hline
0 & 0 & 0 & 0 & 0 & 0 & 0 & 0 & 0 & 0 & 0 & 0 & 0 & 0 & 0 & 0 \\
0 & 0 & 0 & 0 & 0 & 0 & 0 & 0 & 0 & 0 & 0 & 0 & 0 & 1 & 0 & 1 \\
0 & 0 & 0 & 0 & 0 & 0 & 0 & 0 & 0 & 0 & 0 & 0 & 0 & 0 & 0 & 0 \\
1 & 0 & 0 & 0 & 0 & 0 & 0 & 0 & 0 & 0 & 0 & 0 & 1 & 1 & 0 & 1 \\
0 & 0 & 0 & 0 & 0 & 0 & 0 & 0 & 0 & 0 & 0 & 0 & 0 & 0 & 0 & 0 \\
1 & 0 & 0 & 0 & 0 & 0 & 0 & 0 & 0 & 0 & 0 & 0 & 1 & 0 & 0 & 0 \\
0 & 0 & 0 & 0 & 0 & 0 & 0 & 0 & 0 & 0 & 0 & 0 & 0 & 0 & 1 & 1 \\
1 & 0 & 0 & 0 & 0 & 0 & 0 & 0 & 0 & 0 & 0 & 0 & 0 & 1 & 1 &1 \\
1 & 0 & 0 & 0 & 0 & 0 & 0 & 0 & 0 & 0 & 0 & 0 & 1 & 0 & 0 & 0 \\
0 & 1 & 1 & 0 & 0 & 0 & 0 & 0 & 0 & 0 & 0 & 0 & 1 & 1 & 0 & 0 \\
0 & \textcolor{red}{0} & \textcolor{red}{0} & 0 & 0 & 0 & 0 & 0 & 0 & 0 & 0 & 0 & 0 & 0 & 0 & 0 \\

\hline
0 & \textcolor{red}{1} & \textcolor{red}{1} & 0 & 0 & 0 & 0 & 0 & 0 & 0 & 0 & 0 & 0 & 0 & 0 & 0 \\
\end{array}
\right].
$}
\]
Since $A$ and $B$ differ in only four entries, it is straightforward to verify that both belong to the same class of criss-cross codes constructed in Section VI of~\cite{2021TIT_crisscross1}. If we delete the penultimate row and the first column of $A$, and the last row and the first column of $B$, the resulting arrays are identical. Consequently, it is impossible to determine which original array was used based on the received array. 

This example demonstrates a specific $(1,1)$-criss-cross deletion pattern for which the decoding procedure in Section~VI of~\cite{2021TIT_crisscross1} does not recover the original array, and hence the construction in that section does not satisfy the definition of a \((1,1)\)-criss-cross deletion correcting code.

By the same reasoning, the constructions with explicit encoding algorithms presented in Section~VII of~\cite{2021TIT_crisscross1} also do not satisfy the definition of a \((1,1)\)-criss-cross deletion correcting code in certain cases.





\bibliographystyle{IEEEtran}
\bibliography{ref}

@ARTICLE{2024Differential_VT_codes,
  author={Thanh Nguyen, Tuan and Cai, Kui and Siegel, Paul H.},
  journal={IEEE Trans. Inf. Theory}, 
  title={A new version of {$q$}-ary {Varshamov--Tenengolts} codes with more efficient encoders: The differential {VT} codes and the differential shifted {VT} codes},
  year={2024},
  volume={70},
  number={10},
  pages={6989-7004},
  doi={10.1109/TIT.2024.3417894}}

@ARTICLE{2021TIT_crisscross1,
  author={Bitar, Rawad and Welter, Lorenz and Smagloy, Ilia and Wachter-Zeh, Antonia and Yaakobi, Eitan},
  journal={IEEE Trans. Inf. Theory}, 
  title={Criss-cross insertion and deletion correcting codes},

  year={2021},
  volume={67},
  number={12},
  pages={7999-8015},
  doi={10.1109/TIT.2021.3111450}}

@INPROCEEDINGS{2021ISIT_crisscross2,
  author={Chee, Yeow Meng and Hagiwara, Manabu and Van Khu, Vu},
  booktitle={2021 IEEE Int. Symp. Inf. Theory (ISIT)}, 
  title={Two-dimensional deletion correcting codes and their applications},

  year={2021},
  volume={},
  number={},
  pages={2792-2797},
  doi={10.1109/ISIT45174.2021.9517903}}

@article{2025arxiv_crisscross3,
   title={Criss-cross deletion correcting codes: Optimal constructions with efficient decoders},

   journal={https://arxiv.org/abs/2506.07607},
   author={Yubo Sun and Gennian Ge},
   year={2025}, }

@ARTICLE{Existing1,
  author={Roth, R.M.},
  journal={IEEE Trans. Inf. Theory}, 
  title={{Maximum-rank array codes and their application to crisscross error correction}}, 
  year={1991},
  volume={37},
  number={2},
  pages={328-336},
  doi={10.1109/18.75248}}

@article{Existing2,
author = {Gabidulin, Ernst and Pilipchuk, Nina},
year = {2008},
month = {12},
pages = {105-122},
title = {{Error and erasure correcting algorithms for rank codes}},
volume = {49},
journal = {Des. Codes Cryptography},
doi = {10.1007/s10623-008-9185-7}
}

@INPROCEEDINGS{Existing3,
  author={Lund, D. and Gabidulin, E.M. and Honary, B.},
  booktitle={2000 IEEE Int. Symp. Inf. Theory}, 
  title={{A new family of optimal codes correcting term rank errors}}, 
  year={2000},
  volume={},
  number={},
  pages={115},

  doi={10.1109/ISIT.2000.866405}}

@ARTICLE{Existing4,
  author={Blaum, M. and Bruck, J.},
  journal={IEEE Trans. Inf. Theory}, 
  title={{MDS array codes for correcting a single criss-cross error}}, 
  year={2000},
  volume={46},
  number={3},
  pages={1068-1077},
  doi={10.1109/18.841187}}

@ARTICLE{Existing5,
  author={Roth, R.M.},
  journal={IEEE Trans. Inf. Theory}, 
  title={{Probabilistic crisscross error correction}}, 
  year={1997},
  volume={43},
  number={5},
  pages={1425-1438},
  doi={10.1109/18.623142}}

@INPROCEEDINGS{Existing6,
  author={Wachter-Zeh, Antonia},
  booktitle={2014 IEEE Int. Symp. Inf. Theory}, 
  title={{List decoding of crisscross error patterns}}, 
  year={2014},
  volume={},
  number={},
  pages={1236-1240},
  doi={10.1109/ISIT.2014.6875030}}

@ARTICLE{Existing7,
  author={Wachter-Zeh, Antonia},
  journal={IEEE Trans. Inf. Theory}, 
  title={List decoding of criss-cross errors},

  year={2017},
  volume={63},
  number={1},
  pages={142-149},
  doi={10.1109/TIT.2016.2622283}}

@ARTICLE{Existing8,
  author={Krishnamurthy, Akshay and Mazumdar, Arya and McGregor, Andrew and Pal, Soumyabrata},
  journal={IEEE Trans. Inf. Theory}, 
  title={Trace reconstruction: Generalized and parameterized},

  year={2021},
  volume={67},
  number={6},
  pages={3233-3250},
  doi={10.1109/TIT.2021.3066010}}

@INPROCEEDINGS{Existing9,
  author={Bakırtaş, Serhat and Erkip, Elza},
  booktitle={2021 IEEE Int. Symp. Inf. Theory (ISIT)}, 
  title={Database matching under column deletions},

  year={2021},
  volume={},
  number={},
  pages={2720-2725},
  doi={10.1109/ISIT45174.2021.9518145}}

@article{Existing10,
  title={{Class of correcting codes for errors with a lattice configuration}},
  author={Sidorenko, Vladimir Removich},
  journal={Problemy peredachi informatsii},
  volume={12},
  number={3},
  pages={3--12},
  year={1976},
  publisher={Russian Academy of Sciences, Branch of Informatics, Computer Equipment and~…}
}

@article{Existing11,
  title={{Optimum codes correcting lattice errors}},
  author={Gabidulin, {\`E}rnest Mukhamedovich},
  journal={Problemy peredachi informatsii},
  volume={21},
  number={2},
  pages={103--108},
  year={1985},
  publisher={Russian Academy of Sciences, Branch of Informatics, Computer Equipment and~…}
}

@article{Existing12,
  title={{Multi deletion/substitution/erasure error-correcting codes for information in array design}},
  author={Hagiwara, Manabu},
  journal={IEICE Trans. Fundam. Electron. Commun. Comput. Sci.},
  volume={106},
  number={3},
  pages={368--374},
  year={2023},
  publisher={The Institute of Electronics, Information and Communication Engineers}
}

@INPROCEEDINGS{Existing_Crisscross1,
  author={Bitar, Rawad and Smagloy, Ilia and Welter, Lorenz and Wachter-Zeh, Antonia and Yaakobi, Eitan},
  booktitle={2020 Int. Symp. Inf. Theory Appl.
 (ISITA)}, 
  title={Criss-cross deletion correcting codes},

  year={2020},
  volume={},
  number={},
  pages={304-308},
  doi={}}

@INPROCEEDINGS{Existing_Crisscross2,
  author={Hagiwara, Manabu},
  booktitle={2020 Proc. Int. Symp. Inf. Theory Appl. (ISITA)}, 
  title={Conversion method from erasure codes to multi-deletion error-correcting codes for information in array design},

  year={2020},
  volume={},
  number={},
  pages={274-278},
  doi={}}

@ARTICLE{QR_1,
  author={Abdel-Ghaffar, K.A.S. and McEliece, R.J. and Van Tilborg, H.C.A.},
  journal={IEEE Trans. Inf. Theory}, 
  title={{Two-dimensional burst identification codes and their use in burst correction}}, 
  year={1988},
  volume={34},
  number={3},
  pages={494-504},
  doi={10.1109/18.6029}}

@INPROCEEDINGS{QR_2,
  author={Yaakobi, Eitan and Etzion, Tuvi},
  booktitle={2010 IEEE Int. Symp. Inf. Theory}, 
  title={{High dimensional error-correcting codes}}, 
  year={2010},
  volume={},
  number={},
  pages={1178-1182},
  doi={10.1109/ISIT.2010.5513662}}

@article{DNA1,
  title={{A characterization of the DNA data storage channel}},
  author={Heckel, Reinhard and Mikutis, Gediminas and Grass, Robert N},
  journal={Scientific reports},
  volume={9},
  number={1},
  pages={9663},
  year={2019},
  publisher={Nature Publishing Group UK London}
}

@article{rack1,
  title={{Coding for racetrack memories}},
  author={Chee, Yeow Meng and Kiah, Han Mao and Vardy, Alexander and Vu, Van Khu and Yaakobi, Eitan},
  journal={IEEE Trans. Inf. Theory},
  volume={64},
  number={11},
  pages={7094--7112},
  year={2018},
  publisher={IEEE}
}

@inproceedings{rack2,
  title={{Codes correcting limited-shift errors in racetrack memories}},
  author={Chee, Yeow Meng and Kiah, Han Mao and Vardy, Alexander and Van Vu, Khu and Yaakobi, Eitan},
  booktitle={2018 IEEE Int. Symp. Inf. Theory (ISIT)},
  pages={96--100},
  year={2018},
  organization={IEEE}
}

@article{rack3,
  title={{Magnetic domain-wall racetrack memory}},
  author={Parkin, Stuart SP and Hayashi, Masamitsu and Thomas, Luc},
  journal={science},
  volume={320},
  number={5873},
  pages={190--194},
  year={2008},
  publisher={American Association for the Advancement of Science}
}

@inproceedings{rack4,
  title={{Hi-fi playback: Tolerating position errors in shift operations of racetrack memory}},
  author={Zhang, Chao and Sun, Guangyu and Zhang, Xian and Zhang, Weiqi and Zhao, Weisheng and Wang, Tao and Liang, Yun and Liu, Yongpan and Wang, Yu and Shu, Jiwu},
  booktitle={Proceedings of the 42nd Annual Int. Symp. on Computer Architecture},
  pages={694--706},
  year={2015}
}

@article{2022TIT_multi,
  title={{Multiple criss-cross insertion and deletion correcting codes}},
  author={Welter, Lorenz and Bitar, Rawad and Wachter-Zeh, Antonia and Yaakobi, Eitan},
  journal={IEEE Trans. Inf. Theory},
  volume={68},
  number={6},
  pages={3767--3779},
  year={2022},
  publisher={IEEE}
}

@article{1965VTcode,
  title={{A code for correcting a single asymmetric error}},
 author={Varshamov, Rom R and Tenenholtz, G M},
  journal={Automatica i Telemekhanika},
  volume={26},
  number={2},
  pages={288--292},
  year={1965}
}

\end{document}